\documentclass{article}

\usepackage{amsmath, amsthm, amssymb, dsfont, mathtools}
\usepackage{natbib}
\usepackage{enumitem}
\usepackage{xcolor}
\usepackage{hyperref}
\usepackage[margin=1.3in]{geometry}
\hypersetup{colorlinks,citecolor=blue,urlcolor=blue,linkcolor=blue}

\usepackage{algorithm}
\usepackage{algpseudocode}

\usepackage{booktabs}
\usepackage{multirow}

\theoremstyle{plain}
\newtheorem{theorem}{Theorem}[section]
\newtheorem{proposition}[theorem]{Proposition}

\newtheorem{lemma}[theorem]{Lemma}

\theoremstyle{definition}
\newtheorem{example}[theorem]{Example}
\newtheorem{assumption}[theorem]{Assumption}
\newtheorem{remark}[theorem]{Remark}
\newtheorem{definition}[theorem]{Definition}

\newcommand{\E}{\mathbb{E}}    
\newcommand{\R}{\mathbb{R}}    
 
\newcommand{\dd}{\mathrm{d}}     

\renewcommand{\epsilon}{\varepsilon}  

\newcommand{\cD}{\mathcal{D}}

\newcommand{\simiid}{\,{\buildrel \text{iid} \over \sim\,}}

\usepackage{accents}
\newcommand{\ubar}[1]{\underaccent{\bar}{#1}}

\usepackage{nameref}

\newcommand{\id}{\mathds{1}}

\newcommand{\p}{\mathbb{P}}

\newcommand{\truep}{\mathbb{P}^*}

\newcommand{\cP}{\mathcal{P}}

\date{October, 2026}

\title{Tiny but uniform improvements of adaptive BH procedures via compound e-values}

\author{
Nikolaos Ignatiadis\thanks%
 {Department of Statistics and Data Science Institute, University of Chicago.
 E-mail: \href{mailto:ignat@uchicago.edu}{ignat@uchicago.edu}.}
 \and
 Ruodu Wang\thanks%
 {Department of Statistics and Actuarial Science,
 University of Waterloo.
 E-mail: \href{mailto:wang@uwaterloo.ca}{wang@uwaterloo.ca}.}
\and  
Aaditya Ramdas\thanks%
    {Departments of Statistics,
    Stanford University.
    E-mail: \href{mailto:aramdas@stanford.edu}{aramdas@stanford.edu}.}
}
\begin{document}

\maketitle

\begin{abstract}
After the seminal Benjamini--Hochberg (BH) procedure for controlling the false discovery rate (FDR) was proposed, dozens of papers have attempted to improve its power by \emph{adapting to the unknown proportion of nulls}.
We observe that most null-proportion estimates are simply \emph{compound e-values} in disguise, and thus most adaptive FDR procedures can be interpreted as instances of the e-weighted BH (ep-BH) procedure of~\cite*{ignatiadis2024evalues}, i.e., the  BH procedure weighted by compound e-values. This lens helps us show that most existing procedures are inadmissible, and we provide uniform improvements to them. While the improvements are small in practice, they still come for free (without additional assumptions), and help unify the literature. We also use our ``leave-one-out ep-BH method'' to design a new method with finite-sample FDR control for the simultaneous t-test setting.
\end{abstract}

\section{Introduction} 

Compound e-values generalize the notion of e-values and have played a central role in several recent advances in multiple testing. 
Compound e-values can either be used directly for false discovery rate control through the e-BH procedure of~\citet{wang2022false}, or by weighting p-values through the ep-BH procedure of~\citet{ignatiadis2024evalues}.
As one example of the utility of these observations,~\citet{ren2024derandomised} used compound e-values to
derandomize the model-X knockoff filter of~\citet{ candes2018panning}. \cite{ignatiadis2024asymptotic} also showed that \emph{every} FDR controlling procedure can be written as an instance of e-BH with appropriate compound e-values. This latter fact has utility in combining the results from different (possibly randomized) procedures.

\smallskip

In this paper, we show that compound e-values also underlie a much more classical idea in multiple testing: null-proportion adaptivity. That is, most null-proportion estimates are compound e-values in disguise and adaptive multiple testing procedures are instances of ep-BH. This finding further underscores the central role of compound e-values, providing a unifying view on the vast literature on adaptive procedures. Perhaps more importantly, it also allows us to strictly improve the power of existing procedures (e.g., Storey's procedure), and to propose new procedures. 

Table~\ref{tab:procedures} summarizes the procedures considered in this work, that is, procedures that we reinterpret or uniformly improve upon, as well as new procedures we propose. We openly admit that some of these power improvements are expected to be tiny in practice, but they are obtained ``for free'' (without any additional assumptions). In that vein, the title of our paper was inspired by~\cite{solari2017minimally}, who proposed a tiny but uniform improvement to the BH procedure: the minimally adaptive BH (MABH) procedure, which we further discuss in Example~\ref{example:mabh}. 
{
Our uniform improvements are relative to the corresponding adaptive procedures and do not  uniformly improve upon the BH procedure itself.
}

\begin{table}
\centering
\caption{Summary of procedures and their uniform improvements.
}
\label{tab:procedures}
\begin{tabular}{lccc}
\toprule
\multirow{ 2}{*}{Original procedure} & Compound  & Uniform & \multirow{ 2}{*}{Details} \\
& e-values & improvement 
\\
\midrule
Storey~\citep{storey2004strong} & \eqref{eq:ek_storey} & \checkmark & \multirow{ 9}{*}{Section~\ref{subsec:homogeneous_nullprop}}\\
MPC~\citep{dohler2023unified} & \eqref{eq:ek_pounds_cheng} & \checkmark &  \\
MD~\citep{dohler2023unified} & \eqref{eq:doehler_meah} & \checkmark &  \\
Quant~\citep{benjamini2006adaptive} & \eqref{eq:ek_quantile} & \checkmark &  \\ 
IBHlog~\citep{zeisel2011fdr} & \eqref{eq:Ek_IBHlog} & \checkmark & \\
Min-Storey~\citep{zijun2026minstorey} & \eqref{eq:Ek_MS}  & \checkmark & \\
TSTBR~\citep{blanchard2009adaptive} & \eqref{eq:E_k_tstbr} & \checkmark & \\
TST~\citep{benjamini2006adaptive} & \eqref{eq:E_k_tst} & \checkmark &  \\
MABH~\citep{solari2017minimally} & \eqref{eq:mabh} & $=$ &  \\ \hline
W-Max-Storey~\citep{ramdas2019unified} & \eqref{eq:ek_weighted_storey} & \checkmark & \multirow{ 3}{*}{Section~\ref{subsec:wtd_adaptivity}} \\ 
W-LOO-Storey~\citep{zhao2024censored} &     \eqref{eq:ek_weighted_storey_loo} & \checkmark & \\
W-MD+  (ours) & \eqref{eq:ek_wdm+} & (new) &  \\ \hline 
Combination~\citep{dohler2023unifiedV1} & \eqref{eq:compound_evals_combi} & \checkmark & Section~\ref{sec:combining} \\ \hline
LOO-Var+ (ours) & \eqref{eq:loovar_compound_evals} & (new) & Appendix~\ref{sec:simultaneous_ttests} \\ 
\bottomrule
\end{tabular}
\end{table}

One may ask whether the language of compound e-values is necessary at all to achieve the improvements and derive new procedures, or purely coincidental to the goal. Here we must be honest and transparent: we, the authors, discovered these improvements through the lens of compound e-values. Only after we first interpreted the adaptivity constraint as that of being a compound e-value, did we ask whether we can design better compound e-values or whether existing ones were tight. Given the varied universality results underlying e-values and multiple testing (Theorem 4.4 in~\citet{ignatiadis2024asymptotic}, Theorem 13 in~\citet{xu2026bringinga}), it is perhaps not too surprising that taking such an e-value perspective helped us discover improvements to adaptive p-value procedures as well. Thus, our perspective is that this language is helpful, and that the unfamiliar reader may well find it worth it to learn the new terminology.

\section{Background on compound e-values and ep-BH}

For more details on compound e-values and their usage in multiple testing, see~\cite{ignatiadis2024asymptotic} and~\citet[Chapter 9]{ramdas2025hypothesis}. Below we give the basic definitions needed for our paper.
We observe data $X$ drawn from an unknown distribution $\truep$ on some underlying sample space, and consider testing $K$ hypotheses $H_k: \truep \in \mathcal{P}_k$ for $k \in \mathcal{K} = \{1,\dots,K\}$, where each $\mathcal{P}_k$ is a set of distributions. A p-variable for $H_k$ is a $[0,1]$-valued random variable $P_k=P_k(X)$ satisfying $\p(P_k\le t)\le t$ for all $t \in [0,1]$ and $\p \in \mathcal{P}_k$. An e-variable for $H_k$ is a $[0,\infty]$-valued random variable $E_k=E_k(X)$ satisfying $\mathbb{E}^{\p}[E_k]\le 1$ for all $\p \in \mathcal{P}_k$. We use $\cP$ (without subscript) to denote a set of distributions encoding assumptions on the joint data-generating process. Compound p-variables and e-variables relax the above per-hypothesis conditions for p-variables, respectively e-variables, to a single condition that averages over the true nulls.

\begin{definition}[Compound p-variables and e-variables]
\label{defi:compound_evalues}
Fix the null hypotheses $(\cP_1,\dots,\cP_K)$ and a set $\cP $ of distributions.

\begin{enumerate}[label=(\roman*)]
\item Let $P_1,\dots,P_K$ be nonnegative random variables. We say that 
$P_1,\dots,P_K$ are \emph{compound} p-variables for $(\cP_1,\dots,\cP_K)$ under $\cP$ if 
$$
\sum_{k : \p \in \cP_k} \p(P_k\leq t) \leq Kt \qquad \mbox{for all $t\in (0,1)$ and all $\p\in \mathcal P$.}
$$
\item Let $E_1,\dots,E_K$ be $[0,\infty]$-valued random variables. We say that 
$E_1,\dots,E_K$ are \emph{compound} e-variables for $(\cP_1,\dots,\cP_K)$ under $\cP$ if
$$\sum_{k : \p \in \cP_k} \E^{\p}[E_k] \leq K \qquad \mbox{for all $\p\in \mathcal P$.}$$
\end{enumerate}
In both cases, if $K$ in the right hand side is replaced by $K(1+\varepsilon)$ for $\varepsilon \geq 0$, then we 
speak of $\varepsilon$-approximate compound p-variables and e-variables. (This is identical to $(\varepsilon,\delta)$-approximation in \cite{ignatiadis2024asymptotic} with $\delta=0$.)
\end{definition}

{
We omit ``under $\cP$'' in case $\cP=\bigcup_{k\in\mathcal K}\cP_k$.\footnote{In principle, the restriction to $\cP$ can be absorbed into the hypotheses by replacing each $\cP_k$ by $\cP_k\cap\cP$. Keeping $\cP$ separate is often convenient, however, since the $\cP_k$ can focus on specifying the individual hypotheses. For example, if $X=(X_1,\ldots,X_K)$, write $\mathbb P_k^*$ for the marginal distribution of $X_k$ under the joint distribution $\mathbb P^*$. Each $\cP_k$ may then specify a condition on $\mathbb P_k^*$, while $\cP$ imposes restrictions on the dependence among $X_1,\ldots,X_K$.}
We also omit this qualifier when $\cP$ is clear from context. For instance, $\cP$ is not mentioned explicitly in Section~\ref{sec:adaptive_epbh}; there, it is implicitly taken to be the class of distributions satisfying p-independence (Definition~\ref{definition:pindependence}).
}
The terms ``p-values'' and ``e-values" refer to the realizations of p-variables and e-variables, but we sometimes also treat them as random variables as commonly done in the literature. We also write $\mathcal{N} = \{k \in \mathcal{K}\,:\,  \truep \in \mathcal{P}_k\} \subseteq \mathcal{K},$ to denote the indices of true null hypotheses and $\pi_0=|\mathcal N|/K$ to denote the proportion of null hypotheses.

A multiple testing procedure $\mathcal{D}$ is defined as a function mapping the realized values of statistics (e.g., p-values or compound e-values or both) 
to a subset of indices of rejected hypotheses. For example, for procedures based on p-values, we write $\mathcal{D}(\mathbf{p})$ for the rejection set at the realized p-value vector $\mathbf{p}\in[0,1]^K$. For notational convenience, we also use the random variables as the input of multiple testing procedures. The resulting procedures are the same as the ones formally defined with realized values as the input. The hypotheses rejected by  $\cD$ are called discoveries. 
We write $V = V_{\cD}=|\cD \cap \mathcal N|$ as the number of true null hypotheses that are rejected (i.e., false discoveries), and $R = R_{\cD}=|\cD|$ as the total number of discoveries. In this work, we are specifically interested in procedures controlling the false discovery rate~\citep{benjamini1995controlling}, defined as $\mathrm{FDR} \equiv \mathrm{FDR}(\cD, \mathbb P) =  \mathbb E^{\mathbb P}[ V_{\cD}/ R_{\cD}]$ with the convention $0/0=0$. Throughout, $\alpha \in (0,1)$ denotes a prespecified nominal level at which we wish to control the FDR.

{
For procedures based on p-values, we use the following notions of dominance, admissibility and inadmissibility (see also~\citealp[Section 2.2]{ramdas2026minimally}).

\begin{definition}[Dominance and admissibility]
\label{defi:dominance_admissibility}
Let $S\subseteq[0,1]^K$. Let $\mathcal{D}_1$ and $\mathcal{D}_2$ be two multiple testing procedures based on p-values.
We say that $\mathcal{D}_1$ dominates (or uniformly improves upon) $\mathcal{D}_2$ on $S$ if
$
\mathcal{D}_2(\mathbf{p})\subseteq\mathcal{D}_1(\mathbf{p})$ for every $\mathbf{p}\in S$
with strict inclusion for at least one $\mathbf{p}\in S$.  We omit ``on $S$'' when $S=[0,1]^K$.

Given a class of procedures controlling FDR at level $\alpha$ under $\mathcal{P}$, we say that a procedure in this class is inadmissible on $S$ if another procedure in the same class dominates it on $S$. Otherwise, it is admissible on $S$.

\end{definition}

Several of the multiple testing procedures we consider are parameterized by the nominal FDR control level $\alpha$ and other quantities, such as $\tau$ in Storey's procedure (Example~\ref{example:storey}). In this case, dominance means that the inclusion in Definition~\ref{defi:dominance_admissibility} holds for every choice of the parameters and every $\mathbf{p}\in S$, with strict inclusion for at least one parameter choice and one $\mathbf{p}\in S$. 
}

We next introduce the ep-BH procedure, which will allow us to represent and compare adaptive multiple testing procedures in a unified way.

\begin{definition}[The ep-BH procedure]
\label{defi:weighted-pBH}
Let $P_1,\dots, P_K$ be p-variables for the hypotheses $\cP_1,\dots,\cP_K$ and let $E_1,\dots, E_K$ be compound e-variables. Define $Q_k = P_k / E_k$ with the convention $0/0=0$. 
For $k\in \mathcal K$, let $Q_{(k)}$ be the $k$-th order statistic of $Q_1,\ldots,Q_K$, from the smallest to the largest. The ep-BH procedure at level $\alpha$ rejects hypothesis $H_k$ if $Q_k$ is among the smallest $k_\alpha^*$ values $Q_1,\dots,Q_K$, where 
\begin{equation*} 
k_\alpha^*:=\max\left\{k\in \mathcal K: \frac{K Q_{(k)}}{k} \le \alpha \right\},
\end{equation*}     
with the convention $\max(\varnothing) = 0$. 
\end{definition} 

When $E_k=1$ for all $k \in \mathcal{K}$, then the above procedure is the p-BH procedure of~\citet{benjamini1995controlling}. When the $E_k$ are deterministic and such that $\sum_{k \in \mathcal{K}} E_k = K$, then the above procedure is the weighted p-BH procedure of~\citet{genovese2006false}. When $P_k=1$ for all $k \in \mathcal{K}$ and $E_1,\dots,E_K$ are compound e-variables, then the procedure is the e-BH procedure of~\citet{wang2022false}.
In the general case, where $E_1,\dots,E_K$ are compound e-variables, then the above procedure is the ep-BH (e-weighted p-BH) procedure of~\citet{ignatiadis2024evalues}.

In common cases, we can also interpret ep-BH as being p-BH applied to certain compound p-variables.

\begin{proposition}
\label{prop:ep-to-p}
Let $E_1,\ldots,E_K$ be compound e-variables and $P_1,\ldots,P_K$ be p-variables. Assume that for each $k \in \mathcal{N}$, $P_k$ is independent of $E_k$. Then
$Q_k := P_k/E_k$
are compound p-variables. Without independence of $P_k$ and $E_k$ for $k \in \mathcal{N}$, this conclusion can fail.
\end{proposition}
There are two uses of this result.
  First, given that e-variables can be used as unnormalized weights in multiple testing procedures~\citep{ignatiadis2024evalues}, the above proposition actually also provides an interesting example of how compound p-values may naturally arise in multiple testing, different from the examples in~\citet{armstrong2022false, ignatiadis2024empirical}, and~\citet{barber2025false}. 
  The second upshot is that whenever we can prove that ep-BH controls the FDR and the independence condition of Proposition~\ref{prop:ep-to-p} holds, then we have also proved a novel guarantee of FDR control for p-BH with compound p-values.
  As we will see below, we can get a lot of mileage from compound p-values arising in this particular fashion, as compared to the existing general results on p-BH with compound p-values proved in these papers.

\section{New results on FDR control with ep-BH}
\label{sec:general_results_epbh}

Let us start by defining two notions of dependence among multiple p-values in a multiple testing setting, p-independence and positive regression dependence on a subset~\citep[Section 4]{finner2009false} (which is slightly weaker than the original definition in \cite{benjamini2001control}). Below, a set $A\subseteq \R^K$ is said to be non-decreasing
if $x\in A$ implies $y\in A$ for all $y\ge x$.

\begin{definition}[P-independence]
\label{definition:pindependence}
The p-values $P_1,\ldots,P_K$ satisfy the p-independence property if the null p-values $(P_k)_{k \in \mathcal{N}}$ are mutually independent, and the null p-values $(P_k)_{k \in \mathcal{N}}$ are independent of the non-null p-values $(P_k)_{k \notin \mathcal{N}}$.
\end{definition}

\begin{definition}[Positive regression dependence on a subset]
\label{definition:prds}
The p-values $P_1,\ldots, P_K$ satisfy positive regression dependence on a subset if for any null index $k\in \mathcal N$ and non-decreasing set $A  \subseteq  \R^K$, the
function $x\mapsto \p\{(P_{\ell})_{\ell \in \mathcal{K}} \in A\mid P_k\le x\}$ is non-decreasing on $ [0,1]$.  
\end{definition}
We note that p-independence implies PRDS. Occasionally we will use the following shorthand notations: $\mathbf{E} = (E_1,\ldots,E_K)$ and $\mathbf{P}=(P_1,\ldots,P_K)$.
Moreover for any $k \in \mathcal{K}$ we write $\mathbf{P}_{-k}=(P_1,\ldots,P_{k-1},P_{k+1},\ldots, P_K)$ and $\mathbf{P}_{k \to c}=(P_1,\ldots,P_{k-1},c, P_{k+1},\ldots, P_K)$ (where we will use $c=0$ in most cases).

We first provide a general result on the ep-BH procedure with compound e-values.

\begin{theorem}
Let $P_1,\ldots,P_K$ be PRDS p-values, and let $E_1,\ldots,E_K$ be compound e-values that are independent of $P_1,\ldots,P_K$.
Then, ep-BH applied to $\mathbf{P}$ and $\mathbf{E}$ at level $\alpha$ controls $\mathrm{FDR}$, that is,
$\mathrm{FDR}_{\textnormal{ep-BH}} \leq \alpha.$
\label{theorem:epbh_fdr}
\end{theorem}
The above is analogous to Theorem 4 of~\citet{ignatiadis2024evalues} who state such a result for e-values rather than
compound e-values. Neither result generalizes the other because their result makes a stronger assumption (e-values) but also proves a tighter upper bound of $\pi_0 \alpha$. Notably in this result, the dependence among the compound e-values can be arbitrary.

We next provide a new high-level guarantee on FDR control for the ep-BH procedure. This guarantee is  different from Theorem~\ref{theorem:epbh_fdr} because the compound e-values and p-values are dependent.

\begin{theorem}
\label{theorem:loo_epbh}
Suppose that conditional on some random variable $U$ that takes values in a set $\mathcal{U} \subseteq \mathbb R^d$ (for $d \in \mathbb N$), $P_1, \ldots, P_K$ are p-values satisfying p-independence. Suppose further that:
\begin{description}[style=nextline]
\item[($\ast$)] For each $k \in \mathcal{K}$, we have $E_k = g_k(\mathbf{P}, U)$ for some  $g_k \colon [0,1]^{K} \times \mathcal{U} \to [0, \infty]$ that is coordinate-wise non-increasing in its first $K$ coordinates, and moreover $E_k \leq \tilde{E}_k = \tilde{h}_k(\mathbf{P}_{-k}, U)$ for some  $\tilde{h}_k \colon [0,1]^{K-1} \times \mathcal{U} \to [0, \infty]$, where $\tilde{E}_1, \ldots, \tilde{E}_K$ are compound e-values.
\end{description}
Then ep-BH applied to $\mathbf{P}$ and $\mathbf{E}$ at level $\alpha$ satisfies
$
\mathrm{FDR}_{\textnormal{ep-BH}} \leq \alpha$.
If $\tilde{E}_1, \ldots, \tilde{E}_K$
are $\varepsilon$-approximate compound e-values (instead of compound e-values), then ep-BH controls $\mathrm{FDR}$ at $\alpha(1+\varepsilon)$.
\end{theorem}
The result with $\varepsilon$-approximate compound e-values can be used to prove asymptotic FDR control 
as in~\citet{ignatiadis2024asymptotic}.

While condition ($\ast$) might seem complicated at first sight, a simple and frequent way in which it is satisfied is that for each $k \in \mathcal{K}$, we have $E_k = h_k(\mathbf{P}_{-k}, U)$ for some function $h_k \colon [0,1]^{K-1} \times \mathcal{U} \to [0, \infty]$ that is coordinate-wise non-increasing in its first $K-1$ coordinates.
A special case is that  $U$ is independent of $\mathbf{P}$, which 
can arise when $\mathbf{E}$ and $\mathbf{P}$
 are computed on different parts of the entire dataset (assumed to contain independent data points).

For most of our results, we will apply this theorem without any $U$ (formally, taking $U$ to be constant). The additional generality is needed for the LOO-Var+ method developed in Appendix~\ref{sec:simultaneous_ttests}.

Theorem~\ref{theorem:loo_epbh_censoring} of Appendix~\ref{sec:tau_censoring} provides a further guarantee on FDR control for ep-BH that dispenses with the monotonicity requirement.

\section{Adaptive multiple testing with ep-BH}
\label{sec:adaptive_epbh}

We now introduce a slew of existing adaptive multiple testing procedures and our improvements.
Throughout this section, we assume that $P_1,\ldots,P_K$ are p-values that satisfy p-independence. Thus, $\cP$ is implicitly restricted to distributions under which p-independence holds. One exception is Example~\ref{example:mabh} for which we only assume PRDS.

\subsection{Homogeneous null-proportion adaptivity}
\label{subsec:homogeneous_nullprop}
In our setting (with p-independent p-values), p-BH controls the FDR at level $\pi_0 \alpha$; recall that $\pi_0 = |\mathcal{N}|/K$ is the proportion of null hypotheses. Starting from~\citet{benjamini2000adaptive}, this has led to a wealth of procedures that seek to improve the power of p-BH by estimating $\pi_0$. A common blueprint is presented in Algorithm~\ref{algo:usual_null_prop}:

\begin{algorithm}[H]
\caption{Adaptive p-BH procedure with $\hat{\pi}_0$}
\begin{algorithmic}[1]
\Require p-values $\mathbf{P}$, nominal level $\alpha$, null-proportion estimator $\hat{\pi}_0$
\State Compute $\hat{\pi}_0 \equiv \hat{\pi}_0(\mathbf{P})$.
\State Apply p-BH at level $\alpha / \hat{\pi}_0$ with p-values $\mathbf{P}$.
\end{algorithmic}
\label{algo:usual_null_prop}
\end{algorithm}

A main takeaway of our paper is that it is fruitful to interpret adaptive procedures of the form in Algorithm~\ref{algo:usual_null_prop} as instantiations of ep-BH with p-values $P_k$ and compound e-values $E_k$,
\begin{equation}
E_k := \frac{1}{ \hat{\pi}_0}\, \text{ for all }\, k \in \mathcal{K}.
\label{eq:null_prop_adaptivity}
\end{equation}
Below we show that this is indeed true by explaining that the $E_k$ in~\eqref{eq:null_prop_adaptivity} are compound e-values and that
condition ($\ast$) of Theorem~\ref{theorem:loo_epbh} holds for several common null-proportion estimators. In some cases, we postpone this verification to Appendix~\ref{sec:omitted_null_prop}.

\begin{example}[Storey]
Fix $\tau \in (0,1)$. The null-proportion adaptive procedure of~\citet{storey2004strong} is equivalent to the ep-BH procedure with
\begin{equation}
E_k^{\mathrm{Storey}} := \frac{K(1-\tau)}{1 + \sum_{\ell \in \mathcal{K}} \id_{\{ P_{\ell} > \tau\}}}.
\label{eq:ek_storey}
\end{equation}
If the constant $1$ in the denominator is replaced by any constant $c \in (0,1)$, then we still get $\varepsilon$-approximate compound e-values with $\varepsilon = 2(1-c)/\{c(|\mathcal{N}|+1)(1-\tau)\}$.
\label{example:storey}
\end{example}

\begin{remark}[$\tau$-censoring]
The original Storey procedure, as developed by
\citet{storey2004strong}, requires one further modification to the p-BH procedure in addition to using $E_k$ in~\eqref{eq:ek_storey}: they do not allow rejecting any hypothesis with $P_k > \tau$.
This second modification is not required~\citep{benjamini2006adaptive}, and is an artifact of the martingale-based proof of~\citet{storey2004strong}. Nevertheless, in other cases, $\tau$-censoring seems to be more fundamental, e.g., when the null-proportion estimator relies on an optional stopping argument~\citep{gao2025adaptive}, is not monotonic in the p-values~\citep{zijun2026minstorey}, or when one designs multiple testing procedures that use side-information~\citep{li2019multiple, ignatiadis2021covariate, zhao2024censored}. We can express $\tau$-censoring directly in terms of ep-BH; see Appendix~\ref{sec:tau_censoring} for details and guarantees of such constructions.
\label{remark:tau-censoring}
\end{remark}

\begin{example}[Modified Pounds-Cheng, MPC]
\label{example:mpc}
The null-proportion adaptive procedure of~\citet{pounds2006robust} with the modification of~\citet{dohler2023unified} is equivalent to ep-BH with
\begin{equation}
E_k^{\mathrm{MPC}} := \frac{K}{2 + 2\sum_{\ell \in \mathcal{K}} P_{\ell}}.
\label{eq:ek_pounds_cheng}
\end{equation}
\end{example}

\begin{example}[Meah and D{\"o}hler, MD]
\label{example:dohler_meah}
Let $\psi_k: [0,1] \to [0,1]$ be non-decreasing with $\nu_k := \int_0^1 \psi_k(u) \dd u > 0$.\footnote{Suppose that for all $\mathbb P \in \mathcal{P}_k$ we have that $P_k$ is stochastically greater than or equal to a further random variable $P_k^0$. Then we can also take $\nu_k := \mathbb E[\psi_k(P_k^0)] \geq \int_0^1 \psi_k(u) \dd u$. 
As noted by~\citet{dohler2023unified}, such a choice can improve power for discrete tests as compared to using the uniform distribution to compute $\nu_k$.
} The unified procedure of~\citet{dohler2023unified} is equivalent to ep-BH with 
\begin{equation}
E_k^{\mathrm{MD}} :=   \frac{K}{\max_{\ell \in \mathcal{K}} (1/\nu_{\ell}) + \sum_{\ell \in \mathcal{K}} \psi_{\ell}(P_{\ell}) /\nu_{\ell}}. 
\label{eq:doehler_meah}
\end{equation}
When all $\psi_{k}$ and $\nu_k$ are the same, say $\psi$ and $\nu$, the $E_k$ above further simplify to
$
E_k^{\mathrm{MD}} = (K \nu)/(1 + \sum_{\ell \in \mathcal{K}} \psi(P_{\ell})).
$
Moreover, the choices $\psi(u) = \id_{\{u > \tau\}}$ and $\psi(u)=u$ yield Storey's, respectively the MPC procedures.
\end{example}

Below we provide Theorem 11 of~\citet{blanchard2009adaptive} (also see~\citet{benjamini2006adaptive} and Theorem 3.3 in~\citet{sarkar2008methods}) and its application to FDR control of all of the above procedures. We make the following critical leave-one-out assumption on $\hat \pi_0$.

\begin{assumption}[LOO compound e-values]
\label{assumption:loo_compound}    
The null-proportion estimator $\hat{\pi}_0$ is coordinate-wise non-decreasing in $\mathbf P$ and  $ E_k:=\hat{\pi}^{-1}_0(\mathbf{P}_{k \to 0})$, $k \in [K],$ are compound e-values, meaning
\begin{equation}
\sum_{k \in \mathcal{N}}  \mathbb E\left[ \frac{1}{ \hat{\pi}_0(\mathbf{P}_{k \to 0})} \right] \leq K.
\label{eq:BR}
\end{equation}
\end{assumption}

\begin{theorem}[\citet{blanchard2009adaptive}]
Let $P_1,\ldots,P_K$ be p-independent p-values and
$\hat \pi_0$ be a null-proportion estimator such that Assumption~\ref{assumption:loo_compound} holds. 
Then the adaptive p-BH procedure with $\hat{\pi}_0$ (Algorithm~\ref{algo:usual_null_prop}) satisfies 
$
\mathrm{FDR} \leq \alpha.
$
\label{theorem:BR}
\end{theorem}
\begin{proof}
We provide a direct proof in terms of Theorem~\ref{theorem:loo_epbh}. Observe that,
$1/\hat{\pi}_0(\mathbf P) \leq 1/ \hat{\pi}_0(\mathbf{P}_{k \to 0}).$
From this and~\eqref{eq:BR} we get that $E_k := 1/\hat{\pi}_0$ are compound e-values and moreover that condition ($\ast$) of Theorem~\ref{theorem:loo_epbh} is satisfied with the choice
$\tilde{h}_k(\mathbf{P}_{-k}) \equiv 1/ \hat{\pi}_0(\mathbf{P}_{k \to 0}).  
$
\end{proof}
\citet{dohler2023unified} use the above result to show that any adaptive procedure of the unified form in Example~\ref{example:dohler_meah} controls the FDR (under the conditions of the above theorem). To do so, in the proof of their Proposition 3.2, they show that,
\begin{equation}
\mathbb E\left[ \frac{1}{ \hat{\pi}_0(\mathbf{P}_{k \to 0})}\right] \leq \frac{K}{|\mathcal{N}|} \;\text{ for all }\; k \in \mathcal{N}.
\label{eq:DM}
\end{equation}
The above condition, of course, implies~\eqref{eq:BR}. 

Returning to the statement of Theorem~\ref{theorem:BR},
and specifically the compound e-value 
condition~\eqref{eq:BR}, we emphasize that compound e-values (albeit not called that) have played a fundamental role in null-proportion adaptive multiple testing from early on. However, it is more interesting to consider what we have gained by reinterpreting adaptive BH as ep-BH with compound e-values. For one, as we will see below, it leads to simpler proofs of FDR control since directly checking the condition~\eqref{eq:BR} is easier than checking the stronger result in~\eqref{eq:DM}. For instance, the proof of~\citet{dohler2023unified} relies on properties of the convex order, while our proof is more elementary. Moreover, the ep-BH perspective directly suggests new multiple testing procedures.

As a first result, we show that Algorithm~\ref{algo:epbh_null_prop} dominates the adaptive p-BH procedure of Algorithm~\ref{algo:usual_null_prop} under the conditions of Theorem~\ref{theorem:BR}. Thus, for p-independent p-values, our construction establishes inadmissibility within the class of null-proportion adaptive procedures considered in Theorem~\ref{theorem:BR}.

\begin{algorithm}[H]
\caption{Improved adaptive p-BH procedure with $\hat{\pi}_0$ via ep-BH}
\begin{algorithmic}[1]
\Require p-values $\mathbf{P}$, nominal level $\alpha$, null-proportion estimator $\hat{\pi}_0$
\State For each $k \in \mathcal{K}$, compute $E_k = 1/\hat{\pi}_0(\mathbf{P}_{k \to 0})$.
\State Apply ep-BH at level $\alpha$ with p-values $\mathbf{P}$ and compound e-values $\mathbf{E}=(E_1,\ldots,E_K)$. 
\end{algorithmic}
\label{algo:epbh_null_prop}
\end{algorithm}

We have the following theorem.

\begin{theorem}
\label{proposition:generic_loo_adaptive}
Let $P_1,\ldots,P_K$ be p-independent p-values and
$\hat \pi_0$ be a null-proportion estimator such that Assumption~\ref{assumption:loo_compound} holds. Then, Algorithm~\ref{algo:epbh_null_prop}, that is, the ep-BH procedure with p-values $P_1,\ldots,P_K$ and compound e-values $E_1,\dots,E_K$ (from Assumption~\ref{assumption:loo_compound})
satisfies $\mathrm{FDR} \leq \alpha$, and it dominates Algorithm~\ref{algo:usual_null_prop}.
\end{theorem}
The leave-one-out method of Algorithm~\ref{algo:epbh_null_prop} even (minimally) improves 
upon Storey's procedure by using the following compound e-values instead of~\eqref{eq:ek_storey}, $E_k := K(1-\tau)/(1 + \sum_{\ell \neq k} \id_{\{ P_{\ell} > \tau\}}).$ This new procedure can make more discoveries than Storey's procedure in Example~\ref{example:storey} in case a hypothesis with $P_k > \tau$ is rejected. (If we enforce $\tau$-censoring of Remark~\ref{remark:tau-censoring} as in the original paper of~\citet{storey2004strong}, then there is no power-gain to be had, but luckily we do not need to.) Checking the compound e-value condition, we obtain an explicit upper bound strictly below $K$ for $\sum_{k\in\mathcal N}\mathbb E[E_k]$, allowing us to increase these compound e-values further. This leads to a still more powerful procedure that we call Storey+. Here and elsewhere, the sign $+$ denotes a new procedure that uniformly improves upon an existing method.

\begin{proposition}[Storey+]
Let $P_1,\ldots,P_K$ be p-independent p-values and let $\tau \in (0,1)$. Define $E_k$ as:
\begin{equation*}
E_k^{\mathrm{Storey+}} := \frac{1}{1-\tau^K} \cdot \frac{K(1-\tau)}{1 + \sum_{\ell \neq k} \id_{\{ P_{\ell} > \tau\}}}.
\label{eq:loo_storey_plus}
\end{equation*}
Then $E_1,\ldots,E_K$ are compound e-values that satisfy condition ($\ast$) of Theorem~\ref{theorem:loo_epbh} and so the resulting ep-BH controls the $\mathrm{FDR}$. This procedure dominates Storey's procedure in Example~\ref{example:storey} and its leave-one-out modification in Algorithm~\ref{algo:epbh_null_prop}.
\label{proposition:Storey+}
\end{proposition}
Figure~\ref{fig:regions_alpha04_tau02} illustrates the power improvement of Storey+ compared to Storey for $K=2$. However, we acknowledge that power gains would be very small for larger $K$.

\begin{figure}
\centering
\begin{tabular}{cc}
  \small\textbf{(a)} Standard Storey ep-BH, $\alpha=0.4,\ \tau=0.4$ &
  \small\textbf{(b)}  Storey+ (LOO) ep-BH, $\alpha=0.4,\ \tau=0.4$ \\
  \includegraphics[width=0.48\linewidth]{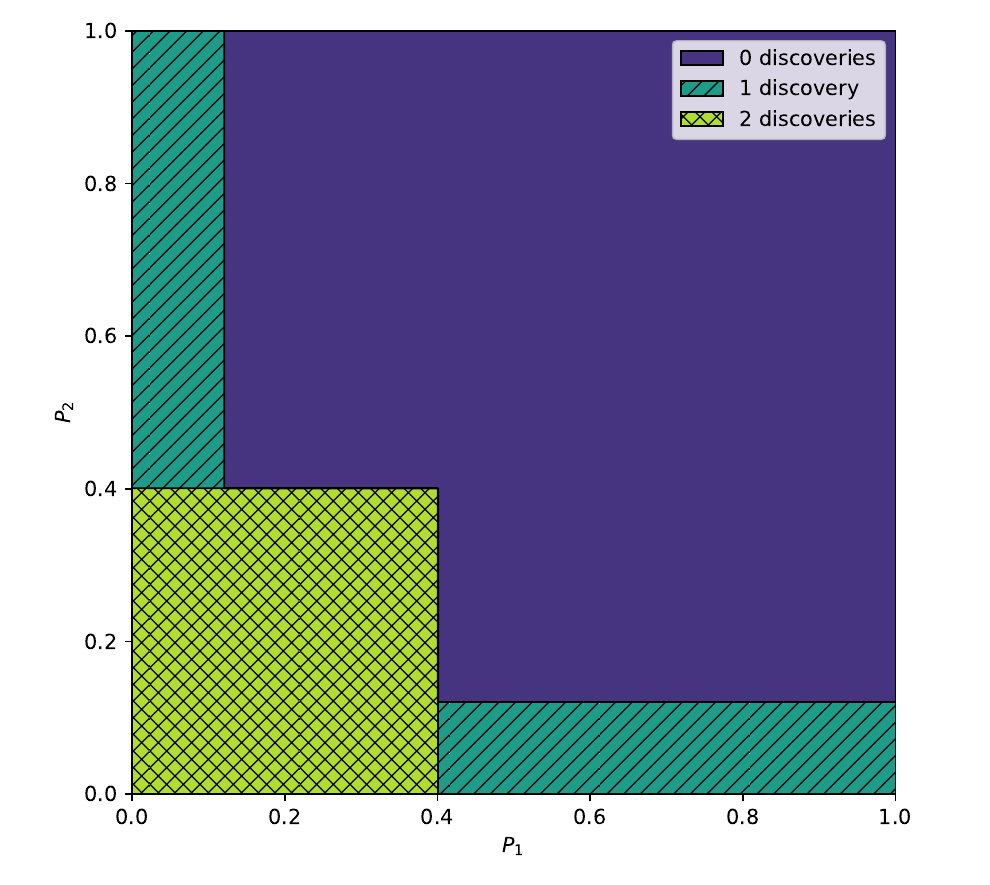} &
  \includegraphics[width=0.48\linewidth]{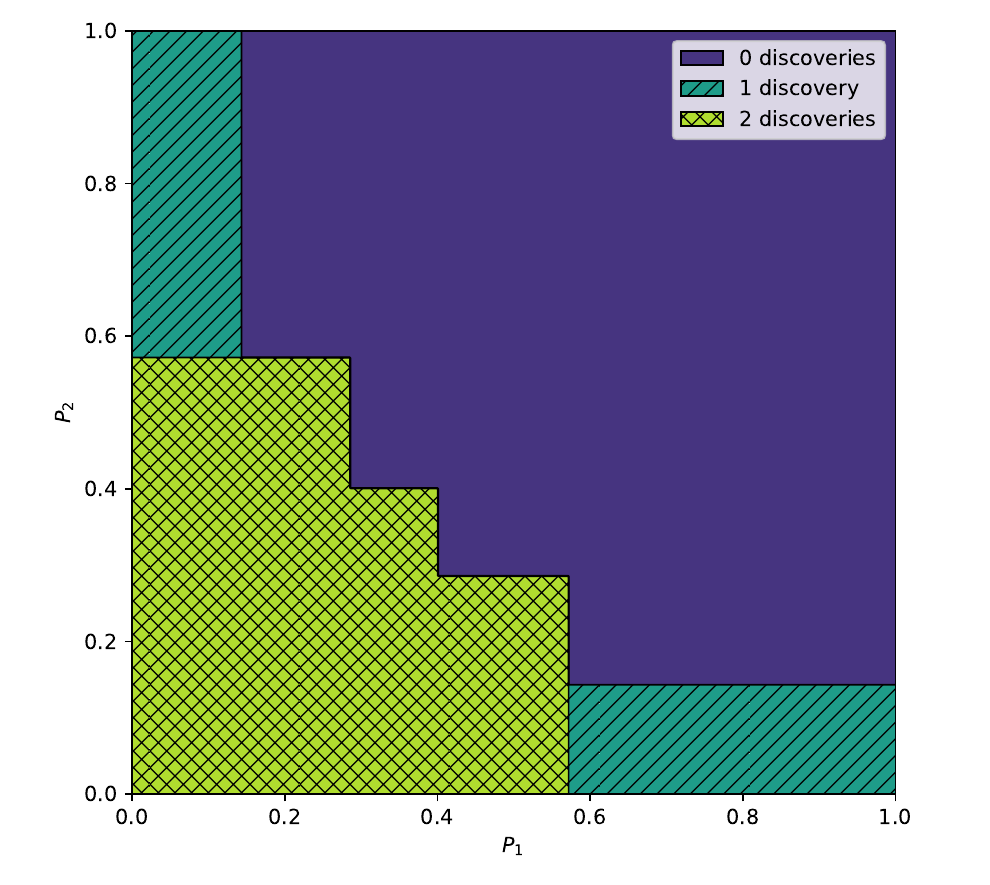}
\end{tabular}

\caption{Rejection regions for two hypotheses ($K=2$). Colors and hatching indicate 0/1/2 discoveries. Our proposed Storey+ dominates standard Storey tuned with the same parameters $\alpha$ and $\tau$.  
}
\label{fig:regions_alpha04_tau02}
\end{figure}

For MPC (Example~\ref{example:mpc}), the same leave-one-out construction yields a procedure called MPC+, which removes the term $2P_k$ from the denominator in~\eqref{eq:ek_pounds_cheng}. We briefly mention three more null-proportion estimators, Quant, IBHlog, and Min-Storey, that can be improved via Algorithm~\ref{algo:epbh_null_prop} by replacing $\mathbf{P}$ by $\mathbf{P}_{k\to0}$ in their implied compound e-values. Their uniform improvements are called Quant+, IBHlog+, and Min-Storey+.

\begin{example}[Quant]
\label{example:quant}
Fix $L \in \{1,\ldots,K\}$. The quantile adaptive procedure of~\citet{benjamini2006adaptive} is equivalent to ep-BH with
\begin{equation}
\label{eq:ek_quantile}
E_k^{\mathrm{Quant}} := \frac{K(1-P_{(L)})}{K-L+1}.
\end{equation}
The implied null-proportion estimator satisfies Assumption~\ref{assumption:loo_compound} for p-independent p-values, see e.g.,~\citet[Section 6.1.5.]{blanchard2009adaptive}.
\end{example}

\begin{example}[IBHlog] 
\label{example:IBHlog}
The IBHlog procedure of~\citet{zeisel2011fdr} is equivalent to ep-BH with
\begin{equation}
\label{eq:Ek_IBHlog}
E_k^{\mathrm{IBHlog}} := \frac{K}{2-\sum_{\ell \in \mathcal{K}} \log(1-P_{\ell})}.
\end{equation}
The implied null-proportion estimator satisfies Assumption~\ref{assumption:loo_compound} for p-independent p-values.
(Note that IBHlog is not a special case of MD since $-\log(1-p)$ is unbounded.)
\end{example}

\begin{example}[Min-Storey] Fix $\epsilon, \ubar{\pi}_0 \in (0,1)$. The Min-Storey procedure of~\citet{zijun2026minstorey} is equivalent to ep-BH with
\begin{equation}
\label{eq:Ek_MS}
E_k^{\text{Min-Storey}} :=   \left(\max\left\{\ubar{\pi}_0,\;\left(C(K, \epsilon, \ubar{\pi}_0) \inf_{\tau \in [0,1-\epsilon]}   \frac{\max\{1,\, \sum_{\ell \in \mathcal{K}}  \id_{\{ P_{\ell} > \tau\}}\}}{K(1-\tau)}\right)\right\}\right)^{-1},
\end{equation}
where $C(K, \epsilon, \ubar{\pi}_0)$ is a constant (given in~\citet[Section 4.3]{zijun2026minstorey}) such that 
Assumption~\ref{assumption:loo_compound} holds for p-independent p-values. 
\citet{zijun2026minstorey} show that Min-Storey is nearly optimal among adaptive p-BH procedures when the null distribution is super-uniform and the alternative is arbitrary (along with some additional regularity conditions).  
\end{example}

We next use our perspective to strengthen the unified adaptive procedure of Example~\ref{example:dohler_meah} even further compared to the strengthening in Theorem~\ref{proposition:generic_loo_adaptive} when $\nu_k$ is heterogeneous for different $k$. We remove the $k$-th summand and replace $\max_{\ell\in\mathcal K}(1/\nu_\ell)$ by $1/\nu_k$. The compound e-value condition permits this smaller term, as shown for W-MD+ below.

\begin{proposition}[MD+]
Let $P_1,\ldots,P_K$ be p-independent p-values. Define $E_k$ as:
$$
E_k^{\mathrm{MD+}} :=  \frac{K}{(1/\nu_{k}) + \sum_{\ell \neq k} \psi_{\ell}(P_{\ell}) /\nu_{\ell}},
$$
where $\nu_k, \psi_k$ are defined as in Example~\ref{example:dohler_meah}.
Then $E_1,\ldots,E_K$ are compound e-values that satisfy condition ($\ast$) of Theorem~\ref{theorem:loo_epbh} and so the resulting ep-BH controls the $\mathrm{FDR}$. This procedure dominates the procedure of~\citet{dohler2023unified} (Example~\ref{example:dohler_meah}) and its leave-one-out modification in Algorithm~\ref{algo:epbh_null_prop}.
\label{proposition:DM+}
\end{proposition}

We end this section with three more procedures that can be written as ep-BH procedures. To this end, it will be useful to introduce the notation $R_{\mathrm{BH},\alpha}(\mathbf{P})$ for the number of rejections of p-BH at level $\alpha$ applied to $\mathbf P$. 

{
\begin{example}[TSTBR]
\label{example:TSTBR}
The two-stage testing (TST) procedure of~\citet{benjamini2006adaptive}, with the modification of~\citet[Section 3.2, estimator $G_3$]{blanchard2009adaptive}, is equivalent to ep-BH with the compound e-values
\begin{equation}
E_k^{\mathrm{TSTBR}} := \frac{1}{1+\alpha}\cdot\frac{K}{1+K-R_{\mathrm{BH},\alpha'}(\mathbf{P})},
\label{eq:E_k_tstbr}
\end{equation}
where $\alpha'=\alpha/(1+\alpha)$. We can improve upon TSTBR using the following compound e-values,
\begin{equation*}
E_k^{\mathrm{TSTBR+}} := \frac{1}{1-(\alpha')^K}\cdot\frac{1}{1+\alpha}\cdot \frac{K}{1+K-R_{\mathrm{BH},\alpha'}(\mathbf{P}_{k \to 0})}.
\end{equation*}
\end{example}

There are two sources of improvement over TSTBR. First, we apply the leave-one-out construction of Theorem~\ref{proposition:generic_loo_adaptive} to the null-proportion estimator implied by~\eqref{eq:E_k_tstbr}. Second, as with Storey+, the resulting compound e-values have unnecessarily low expectation, so we can multiply them by $1/\{1-(\alpha')^K\}$.
}

\begin{example}[TST]
The two-stage testing procedure of~\citet{benjamini2006adaptive} estimates $\hat{\pi}_0 = (K-R_{\mathrm{BH},\alpha'}(\mathbf{P}))/K$ for $\alpha' = \alpha/(1+\alpha)$ and then runs p-BH at level $\alpha' / \hat{\pi}_0$. The TST procedure controls FDR with p-independent p-values and it can be written as ep-BH (at level $\alpha$) with the choice of compound e-values
\begin{equation}
E_k^{\mathrm{TST}} :=  \frac{1}{1+\alpha}\cdot \frac{K}{K-R_{\mathrm{BH},\alpha'}(\mathbf{P}_{k \to 1})}.
\label{eq:E_k_tst}
\end{equation}
Moreover, we can improve upon TST using the following compound e-values,
\begin{equation*}
E_k^{\mathrm{TST+}} :=  \frac{K+\alpha}{K}\cdot \frac{1}{1+\alpha}\cdot \frac{K}{K-R_{\mathrm{BH},\alpha'}(\mathbf{P}_{k \to 1})}.
\end{equation*}
\label{example:TST}
\end{example}

{
For TST, the null-proportion estimator $\hat{\pi}_0$ does not satisfy Assumption~\ref{assumption:loo_compound} even under p-independence.\footnote{{Setting $P_k=0$ guarantees at least one rejection at the first stage. Hence $\hat{\pi}_0(\mathbf{P}_{k \to 0}) \leq (K-1)/K<1$ almost surely for every $k$. Thus under the global null ($|\mathcal{N}|=K$), we have $\sum_{k \in \mathcal{N}} \mathbb E[1/\hat{\pi}_0(\mathbf{P}_{k \to 0})] > K$.}} This was the reason that~\citet{blanchard2009adaptive} introduced the TSTBR procedure of Example~\ref{example:TSTBR},
which is a modification of TST that satisfies Assumption~\ref{assumption:loo_compound} under p-independence. As noted by~\citet[Remark 14]{blanchard2009adaptive}, the original TST procedure dominates TSTBR, so its inadmissibility was already known.
However, \emph{both} TSTBR+ and TST+ improve upon TST, and neither dominates the other (see Appendix~\ref{subsec:comparison_tst}).

The two improvements use different compound e-values. We elaborate on the construction of TST+ because it is unlike our other constructions above. Instead of replacing $P_k$ by $0$, we replace it by $1$ ($\mathbf{P}_{k \to 1}$) to arrive at the compound e-values that reproduce the original TST procedure in~\eqref{eq:E_k_tst}. This leave-one-out construction (setting $P_k$ to $1$) is taken verbatim from the original proof of Theorem 1 of~\citet{benjamini2006adaptive}, who did not interpret this argument in terms of compound e-values.
Since the resulting compound e-values satisfy condition ($\ast$) of Theorem~\ref{theorem:loo_epbh}, TST controls FDR (as already shown by~\citet{benjamini2006adaptive}). At the same time, the resulting compound e-values have unnecessarily low expectation, and so we can multiply them by $(K+\alpha)/K$ to obtain TST+, which dominates TST. 
}

Our final example is the MABH procedure of~\citet{solari2017minimally}, whose paper inspired our title. We show that MABH can also be represented as ep-BH with compound e-values, although here the representation yields neither a new proof of FDR control nor a uniform improvement.

\begin{example}[Minimally adaptive BH, MABH]
\label{example:mabh}
The MABH procedure of~\citet{solari2017minimally} first runs p-BH on $\mathbf{P}$. If p-BH makes no rejections, then MABH also makes no rejections. Otherwise, MABH returns the number of rejections of p-BH applied at level $\alpha K/(K-1)$. MABH may be represented as ep-BH with the following choice of compound e-values,
\begin{equation}
E_k^{\mathrm{MABH}} :=   \frac{K}{K-1} \id_{\{ R_{\mathrm{BH},\alpha}(\mathbf{P})>0\}},
\label{eq:mabh}
\end{equation}
provided that $K\geq 2$ and $\alpha \leq (K-1)/K$.
We note that MABH controls FDR also for PRDS p-values.
\end{example}

{
Theorems~\ref{theorem:epbh_fdr} and~\ref{theorem:loo_epbh} do not apply directly to the e-values in~\eqref{eq:mabh}, even under p-independence, and we still use the argument of~\citet{solari2017minimally} to prove FDR control. In fact, \citet[Theorem 3.1]{ramdas2026minimally} shows that, on $(0,1]^K$, MABH is admissible even under p-independence among deterministic multiple testing procedures whose number of discoveries is coordinate-wise non-increasing in $\mathbf{P}$. Thus, on this domain, we cannot uniformly improve upon MABH using other compound e-values that are coordinate-wise non-increasing functions of $\mathbf{P}$ while retaining FDR control.\footnote{\citet{ramdas2026minimally} shows that MABH is inadmissible if we allow nonmonotone procedures controlling FDR.
We do not pursue uniform improvements using compound e-values that are not coordinate-wise non-increasing in $\mathbf{P}$, since we find the resulting procedures less intuitive.}

\begin{remark}[Global null testing]
As is well known, any procedure controlling FDR at level $\alpha$ also yields a level-$\alpha$ test of the global null: we reject the global null if the multiple testing procedure yields one or more discoveries. MABH and p-BH yield exactly the same global null test (and this property crucially underlies the proof of FDR control of MABH).
Our uniform improvements need not preserve the global null test of the original adaptive procedure; for example, Storey+ can make discoveries even when Storey makes none.
\end{remark}
}

\subsection{Heterogeneous weighted null-proportion adaptivity}
\label{subsec:wtd_adaptivity}

We now assume that for each hypothesis we have a deterministic weight $w_k \in [0,K]$ such that
$\sum_{k \in \mathcal{K}} w_k = K$. A weight $w_k>1$ indicates that we would like to prioritize the $k$-th hypothesis and vice versa.
Null-proportion adaptivity can be combined with such weights, and we now show that the resulting procedures are again instances of ep-BH.

\begin{example}[W-Max-Storey]
\label{example:max_storey}
Fix $\tau \in (0,1)$. The weighted null-proportion adaptive procedure of~\citet{ramdas2019unified}
is equivalent to ep-BH with
\begin{equation}
E_k^{\mathrm{WMaxStorey}} := \frac{Kw_k(1-\tau)}{\max_{\ell \in \mathcal{K}} w_{\ell} + \sum_{\ell \in \mathcal{K}} w_{\ell}\id_{\{ P_{\ell} > \tau\}}}.
\label{eq:ek_weighted_storey}
\end{equation}
\end{example}
\citet{ramdas2019unified} prove that their procedure controls the FDR. One way to prove this result using our framework is by using Theorem~\ref{theorem:loo_epbh} along with condition ($\ast$). The next procedure, due to~\citet{zhao2024censored}, already incorporates a leave-one-out construction in the spirit of Algorithm~\ref{algo:epbh_null_prop}, but enforces $\tau$-censoring (Remark~\ref{remark:tau-censoring}).

\begin{example}[W-LOO-Storey]
\label{loo-storey}
Fix $\tau \in (0,1)$. The weighted null-proportion adaptive procedure of~\citet{zhao2024censored}
is equivalent to ep-BH with
\begin{equation}
E_k^{\mathrm{WStorey}} :=  \frac{Kw_k(1-\tau)}{w_k + \sum_{\ell \neq k} w_{\ell}\id_{\{ P_{\ell} > \tau\}}} \cdot \id_{\{ P_{k} \leq \tau\}}.
\label{eq:ek_weighted_storey_loo}
\end{equation}
\label{example:loo_weighted_storey}
\end{example}
Although not using the terminology of compound e-values, the results of~\citet{zhao2024censored} are closely aligned with the spirit of this paper. They explicitly note the importance of working with compound e-values (their Assumption~3 effectively states the definition of compound e-values), they realize that it is easier to prove FDR control by checking the compound e-value property, and they recognize that adaptive procedures can be strengthened if the null-proportion estimator is different for each hypothesis $k$. 

We can uniformly improve upon both W-Max-Storey and W-LOO-Storey using our compound e-value construction. We call the new procedure W-LOO-Storey+.

\begin{proposition}[W-LOO-Storey+]
Fix $\tau \in (0,1)$. Let $P_1,\ldots,P_K$ be p-independent p-values. Define $E_k$ as:
\begin{equation*}
E_k^{\mathrm{WStorey+}} := \frac{1}{1-\tau^K}\cdot \frac{Kw_k(1-\tau)}{w_k + \sum_{\ell \neq k} w_{\ell}\id_{\{ P_{\ell} > \tau\}}}.
\end{equation*}
Then $E_1,\ldots,E_K$ are compound e-values that satisfy condition ($\ast$) of Theorem~\ref{theorem:loo_epbh} and so the resulting ep-BH controls the $\mathrm{FDR}$ and dominates the procedures of~\citet{ramdas2019unified} and~\citet{zhao2024censored} in Examples~\ref{example:max_storey} and~\ref{example:loo_weighted_storey}.
\label{proposition:wloostorey+} 
\end{proposition}
Compared to W-LOO-Storey, there are two sources of improvement. First, W-LOO-Storey has the additional constraint that no hypothesis with $P_k > \tau$ is allowed to be rejected. We are able to omit this constraint; see Remark~\ref{remark:tau-censoring}. 
This difference can be practically important in the weighted setting: if $w_k$ is large, $P_k / w_k$ could be small, even though $P_k > \tau$. The second source of improvement is the additional multiplicative factor $1/(1-\tau^K)$, analogous to the same factor for Storey+ in Proposition~\ref{proposition:Storey+} which is a special case of the above result. Compared to W-Max-Storey, we also remove the $k$-th summand from the denominator and replace $\max_{\ell\in\mathcal K}w_\ell$ by $w_k$.

Below we use our framework to provide guarantees for a new procedure that is a hybrid of the MD procedure in Example~\ref{example:dohler_meah} and W-LOO-Storey+.

\begin{proposition}[W-MD+]
Let $P_1,\ldots,P_K$ be p-independent p-values. Define $E_k$ as:
\begin{equation}
E_k^{\mathrm{WMD+}} :=  \frac{Kw_k}{(w_k/\nu_{k}) + \sum_{\ell \neq k} \psi_{\ell}(P_{\ell}) (w_{\ell} /\nu_{\ell})},
\label{eq:ek_wdm+}
\end{equation}
where $\nu_k, \psi_k$ are defined as in Example~\ref{example:dohler_meah}.
Then $E_1,\ldots,E_K$ are compound e-values that satisfy condition ($\ast$) of Theorem~\ref{theorem:loo_epbh} and so the resulting ep-BH controls the $\mathrm{FDR}$. 
\label{proposition:dwm+}
\end{proposition}

\begin{proof}
By condition ($\ast$) of Theorem~\ref{theorem:loo_epbh}, it suffices to check that $E_k$ are indeed compound e-values.
Beyond simple manipulations, our arguments only use the facts that $\psi_k(P_k) \in [0,1]$, $\mathbb E[ \psi_{k}(P_k)]/\nu_k \geq 1$ for $k \in \mathcal{N}$ and p-independence:
$$
\begin{aligned}
\sum_{k \in \mathcal{N}} \mathbb E[E_k] &= K\sum_{k \in \mathcal{N}} \mathbb E\left[\frac{w_k}{(w_k/\nu_k) + \sum_{\ell \neq k} \psi_{\ell}(P_{\ell}) (w_{\ell}/\nu_{\ell})}  \right]\\ 
&\leq K\sum_{k \in \mathcal{N}}  \mathbb E\left[\frac{ \psi_k(P_k) (w_k / \nu_k)}{(w_k/\nu_k) + \sum_{\ell \neq k} \psi_{\ell}(P_{\ell}) (w_{\ell}/\nu_{\ell})}  \right] \\ 
&\leq K\sum_{k \in \mathcal{N}}  \mathbb E\left[\frac{ \psi_k(P_k) (w_k / \nu_k)}{\psi_k(P_k)(w_k/\nu_k) + \sum_{\ell \neq k} \psi_{\ell}(P_{\ell}) (w_{\ell}/\nu_{\ell})}  \right] \\ 
&=  K\mathbb E\left[\frac{\sum_{k \in \mathcal{N}}\psi_k(P_k)(w_k/\nu_k)}{\sum_{\ell \in \mathcal{K}} \psi_{\ell}(P_{\ell}) (w_{\ell}/\nu_{\ell})} \right]  \leq K,
\end{aligned}
$$
as desired. The crucial step is the first inequality, where we introduce $\psi_k(P_k)/\nu_k$ into the numerator by taking the expectation over $P_k$ (using p-independence and $\mathbb E[\psi_k(P_k)] \geq \nu_k$ for $k \in \mathcal{N}$). We call this step the \emph{Gavrilov replacement}, as it mimics an argument in~\citet{gavrilov2009adaptive}.

\end{proof}

Taking $w_k=1$ for all $k$ recovers MD+ and proves its compound e-value property. For the choice $\psi_k(u)=\id_{\{u>\tau\}}$ and $\nu_k=1-\tau$, the W-MD+ construction specializes to the unscaled W-LOO-Storey construction. The argument for W-LOO-Storey+ above shows that, in this special case, the resulting compound e-values can be further multiplied by $1/(1-\tau^K)$.

\begin{remark}[Computation for W-MD+] All of $E_1^{\mathrm{WMD+}},\dots,E_K^{\mathrm{WMD+}}$ can be computed in $O(K)$ time: letting $\bar{S} := \sum_{\ell \in \mathcal{K}} \psi_{\ell}(P_{\ell})(w_{\ell}/\nu_{\ell})$, we have $E_k^{\mathrm{WMD+}} = Kw_k / \{\bar{S} + (1-\psi_{k}(P_k))(w_k/\nu_k)\}$. Similar shortcuts are available for other improved methods proposed in this paper. \end{remark}

The above examples of improvements of existing procedures are far from comprehensive; for example, null-proportion adaptivity is also used in other algorithms such as the p-filter~\citep{ramdas2019unified} (in a manner consistent with compound e-variables) and the uniform improvements would thus be inherited by many such downstream procedures.

\section{Combining null-proportion estimators}
\label{sec:combining}
In this section we study adaptive BH procedures that combine two null-proportion estimators. Specifically,
suppose we have access to two null-proportion estimators, $\hat{\pi}_0^{\mathrm{A}}$ and $\hat{\pi}_0^{\mathrm{B}}$.~\citet{dohler2023unifiedV1} propose to combine these estimators by taking a convex combination,\footnote{We cite the first version of the arXiv preprint~\citep{dohler2023unifiedV1}  instead of the published article~\citep{dohler2023unified}, since only the former contains the results on combining null-proportion estimators.
}
\begin{equation} 
\hat{\pi}_0^{\lambda} := \lambda \hat{\pi}_0^{\mathrm{A}} + (1-\lambda)\hat{\pi}_0^{\mathrm{B}},\,\; \lambda \in (0,1),
\label{eq:adaptive_null_prop}
\end{equation}
and then to apply the adaptive BH procedure with $\hat{\pi}_0^{\lambda}$ as the estimate of $\pi_0$. This combination strategy is presented in Algorithm~\ref{algo:basic_combination}:
\begin{algorithm}[H]
\caption{Adaptive p-BH procedure combining $\hat{\pi}_0^{\mathrm{A}}$ and $\hat{\pi}_0^{\mathrm{B}}$}
\begin{algorithmic}[1]
\Require p-values $\mathbf{P}$, nominal level $\alpha$, null-proportion estimators $\hat{\pi}_0^{\mathrm{A}}$ , $\hat{\pi}_0^{\mathrm{B}}$, $\lambda \in (0,1)$
\State Compute $\hat{\pi}_0^{\mathrm{A}} \equiv \hat{\pi}_0^{\mathrm{A}}(\mathbf{P})$ and
$\hat{\pi}_0^{\mathrm{B}} \equiv \hat{\pi}_0^{\mathrm{B}}(\mathbf{P})$.
\State Let $\hat{\pi}_0^{\lambda} := \lambda \hat{\pi}_0^{\mathrm{A}} + (1-\lambda)\hat{\pi}_0^{\mathrm{B}}$ as in~\eqref{eq:adaptive_null_prop}.
\State Apply p-BH at level $\alpha / \hat{\pi}_0^{\lambda}$ with p-values $\mathbf{P}$.
\end{algorithmic}
\label{algo:basic_combination}
\end{algorithm}
As in our general framework, we can interpret Algorithm~\ref{algo:basic_combination} as an instantiation of ep-BH with compound e-values,
\begin{equation}
E_k := \frac{1}{\hat{\pi}_0^{\lambda}} = \frac{1}{\lambda \hat{\pi}_0^{\mathrm{A}} + (1-\lambda)\hat{\pi}_0^{\mathrm{B}}}.
\label{eq:compound_evals_combi}
\end{equation}
We have the following guarantee on Algorithm~\ref{algo:basic_combination}.

\begin{proposition}[Generalization of Proposition 3.3\ in~\citet{dohler2023unifiedV1}]
Let $P_1,\ldots,P_K$ be p-independent p-values and
$\hat{\pi}_0^{\mathrm{A}}$, $\hat{\pi}_0^{\mathrm{B}}$ be two null-proportion estimators such that Assumption~\ref{assumption:loo_compound} holds.
Then the procedure of Algorithm~\ref{algo:basic_combination} 
satisfies $\mathrm{FDR} \leq \alpha$.
\end{proposition}
We provide a short proof of the proposition using our framework since it will enable us to propose a more powerful combination strategy.
\begin{proof}
Let us consider the leave-one-out compound e-values,
\begin{equation}
E_k^{\mathrm{A}} := \frac{1}{\hat{\pi}^{\mathrm{A}}_0(\mathbf{P}_{k \to 0})},\;\;\; E_k^{\mathrm{B}} := \frac{1}{\hat{\pi}^{\mathrm{B}}_0(\mathbf{P}_{k \to 0})}.
\label{eq:loo_compound_A_B}
\end{equation}
The result of the proposition will follow if we can show that 
\begin{equation}
E_k^{\mathrm{DM},\lambda} := \frac{1}{\hat{\pi}_0^{\lambda}(\mathbf{P}_{k \to 0})}, 
\label{eq:harmonic_avg}
\end{equation}
are also compound e-values. By convexity of $x \mapsto 1/x$ on $(0,\infty)$,
\begin{equation}
E_k^{\mathrm{DM},\lambda} \leq \frac{\lambda}{ \hat{\pi}_0^{\mathrm{A}}(\mathbf{P}_{k \to 0})} + \frac{1-\lambda}{\hat{\pi}_0^{\mathrm{B}}(\mathbf{P}_{k \to 0})} = \lambda E_k^{\mathrm{A}} + (1-\lambda) E_k^{\mathrm{B}},
\label{eq:combination-evalue-upper-bound}
\end{equation}
and so indeed $E_k^{\mathrm{DM},\lambda}$ are compound e-values. We conclude by Theorem~\ref{theorem:loo_epbh} after checking that the coordinate-wise monotonicity requirement is satisfied.
\end{proof}
From \eqref{eq:combination-evalue-upper-bound} we see that there are two possible improvements over Algorithm~\ref{algo:basic_combination}. First, we can use the LOO-improvement from Algorithm~\ref{algo:epbh_null_prop}. Second, 
it is preferable (i.e., strictly better) to average the compound e-values via
\begin{equation} 
E_k^{\lambda} := \lambda E_k^{\mathrm{A}} + (1-\lambda)E_k^{\mathrm{B}},\,\; \lambda \in (0,1),
\label{eq:combine_combounds_evals}
\end{equation}
instead of averaging the null-proportion estimators. The reason is that averaging the null-proportion estimators means that we are taking the harmonic average of the implied compound e-values, which is smaller than their arithmetic average. We propose the following general procedure:
\begin{algorithm}[H]
\caption{Improved adaptive p-BH procedure combining $\hat{\pi}_0^{\mathrm{A}}$ and $\hat{\pi}_0^{\mathrm{B}}$ via ep-BH}
\begin{algorithmic}[1]
\Require p-values $\mathbf{P}$, nominal level $\alpha$, null-proportion estimators $\hat{\pi}_0^{\mathrm{A}}$ , $\hat{\pi}_0^{\mathrm{B}}$, $\lambda \in (0,1)$
\State For each $k \in \mathcal{K}$, compute $E_k^{\mathrm{A}}= 1/\hat{\pi}_0^{\mathrm{A}}(\mathbf{P}_{k\to 0})$ and $E_k^{\mathrm{B}}= 1/\hat{\pi}_0^{\mathrm{B}}(\mathbf{P}_{k\to 0})$.
\State For each $k \in \mathcal{K}$, let $E_k^{\lambda} := \lambda E_k^{\mathrm{A}} + (1-\lambda)E_k^{\mathrm{B}}$ as in~\eqref{eq:combine_combounds_evals}.
\State Apply ep-BH at level $\alpha$ with p-values $\mathbf{P}$ and compound e-values $\mathbf{E}^{\lambda}=(E_1^{\lambda},\ldots,E_K^{\lambda})$. 
\end{algorithmic}
\label{algo:epbh_combination}
\end{algorithm}

We have the following result for Algorithm~\ref{algo:epbh_combination}.
\begin{proposition}
Let $P_1,\ldots,P_K$ be p-independent p-values and
$\hat{\pi}_0^{\mathrm{A}}$, $\hat{\pi}_0^{\mathrm{B}}$ be two null-proportion estimators such that Assumption~\ref{assumption:loo_compound} holds.
Then the procedure of Algorithm~\ref{algo:epbh_combination} satisfies $\mathrm{FDR} \leq \alpha$, and it dominates Algorithm~\ref{algo:basic_combination}.
\label{prop:new_combi}
\end{proposition}

\section{Discussion}
\label{sec:discussion}

In this paper, we have shown that most null-proportion adaptive procedures are instances of ep-BH with appropriate compound e-values. In this final section, we (i) recap the general recipe underlying our uniform improvements, (ii) point to a new procedure for simultaneous t-tests together with accompanying numerical results, highlighting further possibilities opened up by our perspective, and (iii) offer an outlook on the role of compound e-values in multiple testing.

\paragraph{Constructing larger compound e-values.} We have shown that most null-proportion adaptive procedures are instances of ep-BH with appropriate compound e-values. Since ep-BH makes more discoveries whenever the compound e-values are larger, our improvements amount to the following recipe: take the compound e-values implied by an existing procedure and construct larger ones that still satisfy condition ($\ast$) of Theorem~\ref{theorem:loo_epbh}. Looking back at Table~\ref{tab:procedures}, this recipe was carried out in three ways. The first is the leave-one-out construction: replacing $1/\hat{\pi}_0(\mathbf{P})$ by the larger $1/\hat{\pi}_0(\mathbf{P}_{k \to 0})$, which works for any estimator satisfying Assumption~\ref{assumption:loo_compound} (Theorem~\ref{proposition:generic_loo_adaptive}). The second is realizing that the compound e-values have unnecessarily low expectation and can therefore be increased. This is how we obtained the additional multiplicative factors in Storey+, TSTBR+, TST+ and W-LOO-Storey+, and the hypothesis-specific term $w_k/\nu_k$ in W-MD+ (Proposition~\ref{proposition:dwm+}). The third is averaging on the right scale: when combining null-proportion estimators, it is better to average the compound e-values than the estimators, since the latter is a harmonic (rather than arithmetic) average of the former (Proposition~\ref{prop:new_combi}).

\paragraph{New methods for simultaneous t-tests and numerical results.} Beyond improving existing methods, the compound e-value perspective can also lead to new constructions. In Appendix~\ref{sec:simultaneous_ttests}, we consider the problem of simultaneous t-tests, where for each hypothesis we observe $n$ independent normal replicates with unknown mean and variance and we seek to test whether the means are equal to zero. Building on a line of work that leverages information beyond the t-test p-values~\citep{westfall2004weighted, finos2007fdr, bourgon2010independent}, we construct compound e-values by applying weighting functions to statistics that are independent of the null p-values (Theorem~\ref{theorem:loo_var}), yielding a new procedure, LOO-Var+. Our procedure attains finite-sample FDR control for any non-decreasing weighting function, in contrast to the asymptotic guarantees of~\citet{ignatiadis2024evalues}. We also provide a simulation study illustrating the power gains of LOO-Var+ and of other procedures considered in this work, such as Algorithm~\ref{algo:epbh_combination}.

\paragraph{The central role of compound e-values.} Stepping back, this paper adds to the evidence that compound e-values are a central concept in multiple testing:~\citet{ignatiadis2024asymptotic} showed that every FDR controlling procedure can be written as the e-BH procedure of~\citet{wang2022false} with appropriate compound e-values, and we have now seen that adaptive procedures have compound e-values built into them. As we admitted in the introduction, we discovered the improvements in this paper only after taking the compound e-value perspective, and we expect that compound e-values will continue to play a central role in the development of multiple testing methodology.

\subsection*{Acknowledgments.} We thank Sebastian D{\"o}hler, Zijun Gao, and Etienne Roquain for helpful feedback on an earlier version of the manuscript.
This work was completed in part with resources provided by the University of Chicago’s Research Computing Center. N.I.~gratefully acknowledges support from the U.S.~National Science Foundation (DMS-2443410). R.W.~is supported by the Natural Sciences and Engineering Research Council of Canada (CRC-2022-00141, RGPIN-2024-03728).

\bibliographystyle{abbrvnat}
\bibliography{references}

@article{xu2026bringinga,
  title = {Bringing Closure to False Discovery Rate Control: A General Principle for Multiple Testing},
  shorttitle = {Bringing {{Closure}} to {{False Discovery Rate Control}}},
  author = {Xu, Ziyu and Solari, Aldo and Fischer, Lasse and de Heide, Rianne and Ramdas, Aaditya and Goeman, Jelle},
  year = {2026},
  journal = {arXiv preprint},
  volume = {arXiv:2509.02517}
}

@article{armstrong2022false,
  title = {False Discovery Rate Adjustments for Average Significance Level Controlling Tests},
  author = {Armstrong, Timothy B.},
  year = {2022},
  eprint = {2209.13686},
  publisher = {arXiv},
  journal = {arXiv preprint},
  volume = {arXiv:2209.13686}
}

@article{ignatiadis2024asymptotic,
  title={Asymptotic and compound e-values: multiple testing and empirical {B}ayes},
  author={Ignatiadis, Nikolaos and Wang, Ruodu and Ramdas, Aaditya},
  journal={arXiv preprint},
  volume = {arXiv:2409.19812},
  year={2025}
}

@article{ren2024derandomised,
  title = {Derandomised Knockoffs: Leveraging {\emph{e}}-Values for False Discovery Rate Control},
  shorttitle = {Derandomised Knockoffs},
  author = {Ren, Zhimei and Barber, Rina Foygel},
  year = {2024},
  journal = {Journal of the Royal Statistical Society Series B},
  volume = {86},
  number = {1},
  pages = {122--154}
}

@article{wang2022false, 
  title = {False Discovery Rate Control with E-values},
  author = {Wang, Ruodu and Ramdas, Aaditya},
  year = {2022},
  journal = {Journal of the Royal Statistical Society Series B},
  volume = {84},
  number = {3},
  pages = {822--852}
}

@article{genovese2006false,
  title = {False Discovery Control with P-Value Weighting},
  author = {Genovese, Christopher R. and Roeder, Kathryn and Wasserman, Larry},
  year = {2006},
  journal = {Biometrika},
  volume = {93},
  number = {3},
  pages = {509--524}
}

@article{benjamini1995controlling,
  title = {Controlling the False Discovery Rate: A Practical and Powerful Approach to Multiple Testing},
  shorttitle = {Controlling the {{False Discovery Rate}}},
  author = {Benjamini, Yoav and Hochberg, Yosef},
  year = {1995},
  journal = {Journal of the Royal Statistical Society Series B},
  volume = {57},
  number = {1},
  pages = {289--300}
}

@article{ramdas2019unified,
  title = {A Unified Treatment of Multiple Testing with Prior Knowledge Using the P-Filter},
  author = {Ramdas, Aaditya and Barber, Rina F. and Wainwright, Martin J. and Jordan, Michael I.},
  year = {2019},
  journal = {The Annals of Statistics},
  volume = {47},
  number = {5},
  pages = {2790--2821}
}

@article{ignatiadis2021covariate,
  title = {Covariate Powered Cross-weighted Multiple Testing},
  author = {Ignatiadis, Nikolaos and Huber, Wolfgang},
  year = {2021},
  journal = {Journal of the Royal Statistical Society Series B},
  volume = {83},
  number = {4},
  pages = {720--751}
}

@article{li2019multiple,
  title = {Multiple Testing with the Structure-Adaptive {{Benjamini}}--{{Hochberg}} Algorithm},
  author = {Li, Ang and Barber, Rina Foygel},
  year = {2019},
  journal = {Journal of the Royal Statistical Society Series B: Statistical Methodology},
  volume = {81},
  number = {1},
  pages = {45--74},
}

@article{benjamini2001control,
  title = {The Control of the False Discovery Rate in Multiple Testing under Dependency},
  author = {Benjamini, Yoav and Yekutieli, Daniel},
  year = {2001},
  journal = {The Annals of Statistics},
  volume = {29},
  number = {4},
  pages = {1165--1188}
}

@article{storey2004strong,
  title = {Strong Control, Conservative Point Estimation and Simultaneous Conservative Consistency of False Discovery Rates: A Unified Approach},
  shorttitle = {Strong Control, Conservative Point Estimation and Simultaneous Conservative Consistency of False Discovery Rates},
  author = {Storey, John D. and Taylor, Jonathan E. and Siegmund, David},
  year = {2004},
  journal = {Journal of the Royal Statistical Society Series B},
  volume = {66},
  number = {1},
  pages = {187--205}
}

@article{ignatiadis2024empirical,
  title = {Empirical Partially {{Bayes}} Multiple Testing and Compound {$\chi^2$} Decisions},
  author = {Ignatiadis, Nikolaos and Sen, Bodhisattva},
  year = {2025},
  journal = {The Annals of Statistics},
  volume = {53},
  number = {1},
  pages = {1--36},
}

@article{candes2018panning,
  title = {Panning for Gold: `{{Model-X}}' Knockoffs for High Dimensional Controlled Variable Selection},
  shorttitle = {Panning for {{Gold}}},
  author = {Cand{\`e}s, Emmanuel and Fan, Yingying and Janson, Lucas and Lv, Jinchi},
  year = {2018},
  journal = {Journal of the Royal Statistical Society Series B},
  volume = {80},
  number = {3},
  pages = {551--577}
}

@article{blanchard2008two,
  title = {Two Simple Sufficient Conditions for {{FDR}} Control},
  author = {Blanchard, Gilles and Roquain, Etienne},
  year = {2008},
  journal = {Electronic Journal of Statistics},
  volume = {2},
  pages = {963--992}
}

@article{ignatiadis2024evalues,
  title = {E-Values as Unnormalized Weights in Multiple Testing},
  author = {Ignatiadis, Nikolaos and Wang, Ruodu and Ramdas, Aaditya},
  year = {2024},
  journal = {Biometrika},
  volume = {111},
  number = {2},
  pages = {417--439}
}

@incollection{westfall2004weighted,
  title = {Weighted {{FWE-controlling}} Methods in High-Dimensional Situations},
  booktitle = {Recent {{Developments}} in {{Multiple Comparison Procedures}}},
  author = {Westfall, Peter H. and Kropf, Siegfried and Finos, Livio},
  year = {2004},
  series = {Institute of {{Mathematical Statistics Lecture Notes}} - {{Monograph Series}}},
  volume = {47},
  pages = {143--154},
  publisher = {{Institute of Mathematical Statistics}},
  address = {{Beachwood, Ohio, USA}}
}

@article{finos2007fdr,
  title = {{{FDR-}} and {{FWE-controlling}} Methods Using Data-Driven Weights},
  author = {Finos, Livio and Salmaso, Luigi},
  year = {2007},
  journal = {Journal of Statistical Planning and Inference},
  volume = {137},
  number = {12},
  pages = {3859--3870}
}

@article{bourgon2010independent,
  title={Independent filtering increases detection power for high-throughput experiments},
  author={Bourgon, Richard and Gentleman, Robert and Huber, Wolfgang},
  journal={Proceedings of the National Academy of Sciences},
  volume={107},
  number={21},
  pages={9546--9551},
  year={2010},
  publisher={National Acad Sciences}
}

@article{guo2017analysis,
  title={Analysis of error control in large scale two-stage multiple hypothesis testing},
  author={Guo, Wenge and Romano, Joseph P},
  journal={arXiv preprint, arXiv:1703.06336},
  year={2017}
}

@article{ramdas2025hypothesis,
  title = {Hypothesis Testing with E-Values},
  author = {Ramdas, Aaditya and Wang, Ruodu},
  year = {2025},
  journal = {Foundations and Trends in Statistics},
  volume = {1},
  number = {1-2},
  pages = {1--390},
}

@article{finner2009false,
  title={On the false discovery rate and an asymptotically optimal rejection curve},
  author={Finner, Helmut and Dickhaus, Thorsten and Roters, Markus},
  journal={The Annals of Statistics},
  volume={37},
  number={2},
  pages={596--618},
  year={2009},
}

@article{zhao2024censored,
  title = {$\tau$-Censored Weighted {{Benjamini}}--{{Hochberg}} Procedures under Independence},
  author = {Zhao, Haibing and Zhou, Huijuan},
  year = {2024},
  journal = {Biometrika},
  volume = {111},
  number = {2},
  pages = {479--496},
}

@article{dohler2023unified,
  title = {Improved Null Proportion Estimators for Multiple Discrete Tests with Plug-in {{FDR}} Control},
  author = {Meah, Iqraa and D{\"o}hler, Sebastian},
  year = {2026},
  journal = {Computational Statistics \& Data Analysis},
  volume = {224},
  pages = {108429},
}

@article{dohler2023unifiedV1,
  title = {A Unified Class of Null Proportion Estimators with Plug-in {{FDR}} Control},
  author = {D{\"o}hler, Sebastian and Meah, Iqraa},
  year = {2023},
  journal = {arXiv preprint (v1)},
  volume = {arXiv:2307.13557v1}
}

@article{barber2025false,
  title={False discovery rate control with compound p-values},
  author={Barber, Rina Foygel and Samworth, Richard J},
  journal={Journal of the Royal Statistical Society: Series B},
  pages={qkag116},
  year={2026},
  publisher={Oxford University Press UK}
}

@article{benjamini2006adaptive,
  title = {Adaptive Linear Step-up Procedures That Control the False Discovery Rate},
  author = {Benjamini, Yoav and Krieger, Abba M. and Yekutieli, Daniel},
  year = {2006},
  journal = {Biometrika},
  volume = {93},
  number = {3},
  pages = {491--507},
}

@article{pounds2006robust,
  title = {Robust Estimation of the False Discovery Rate},
  author = {Pounds, Stan and Cheng, Cheng},
  year = {2006},
  journal = {Bioinformatics},
  volume = {22},
  number = {16},
  pages = {1979--1987},
}

@article{blanchard2009adaptive,
  title = {Adaptive False Discovery Rate Control under Independence and Dependence},
  author = {Blanchard, Gilles and Roquain, Etienne},
  year = {2009},
  journal = {Journal of Machine Learning Research},
  volume = {10},
  number = {97},
  pages = {2837--2871}
}

@article{sarkar2008methods,
  title = {On Methods Controlling the False Discovery Rate},
  author = {Sarkar, Sanat K},
  year = {2008},
  journal = {Sankhy{\=a} A: The Indian Journal of Statistics},
  pages = {135--168},
}

@article{solari2017minimally,
  title = {Minimally Adaptive {{BH}}: {{A}} Tiny but Uniform Improvement of the Procedure of {{Benjamini}} and {{Hochberg}}},
  shorttitle = {Minimally Adaptive {{BH}}},
  author = {Solari, Aldo and Goeman, Jelle J.},
  year = {2017},
  journal = {Biometrical Journal},
  volume = {59},
  number = {4},
  pages = {776--780},
}

@article{gao2025adaptive,
  title = {An Adaptive Null Proportion Estimator for False Discovery Rate Control},
  author = {Gao, Zijun},
  year = {2025},
  journal = {Biometrika},
  volume = {112},
  number = {1},
  pages = {asae051},
}

@article{zeisel2011fdr,
  title = {{{FDR}} Control with Adaptive Procedures and {{FDR}} Monotonicity},
  author = {Zeisel, Amit and Zuk, Or and Domany, Eytan},
  year = {2011},
  journal = {The Annals of Applied Statistics},
  volume = {5},
  number = {2A},
  pages = {943--968},
}

@article{benjamini2000adaptive,
  title = {On the Adaptive Control of the False Discovery Rate in Multiple Testing with Independent Statistics},
  author = {Benjamini, Yoav and Hochberg, Yosef},
  year = {2000},
  journal = {Journal of Educational and Behavioral Statistics},
  volume = {25},
  number = {1},
  pages = {60--83},
}

@article{zijun2026minstorey,
  title = {On Min-{{Storey}} Estimators for Multiple Testing and Conformal Novelty Detection},
  author = {Gao, Zijun and Roquain, Etienne},
  year = {2026},
  journal = {arXiv preprint},
  volume = {arXiv:2603.17984}
}

@article{gavrilov2009adaptive,
  title = {An Adaptive Step-down Procedure with Proven {{FDR}} Control under Independence},
  author = {Gavrilov, Yulia and Benjamini, Yoav and Sarkar, Sanat K.},
  year = {2009},
  journal = {The Annals of Statistics},
  volume = {37},
  number = {2},
  pages = {619--629},
}

@article{ramdas2026minimally,
  title = {Minimally Adaptive {{BH}} Is Admissible, Canonical Closed {{BH}} Is Not},
  author = {Ramdas, Aaditya},
  year = {2026},
  journal = {Stanford Digital Repository},
  doi = {10.25740/TF843RN7918},
}

\appendix 

\setcounter{equation}{0}
\setcounter{figure}{0}
\setcounter{table}{0}

\renewcommand{\theequation}{S\arabic{equation}}
\renewcommand{\thefigure}{S\arabic{figure}}
\renewcommand{\thetable}{S\arabic{table}}

\newpage 
\section{The ep-BH procedure and \texorpdfstring{$\tau$}{tau}-censoring}
\label{sec:tau_censoring}

We first define the $\tau$-censored ep-BH procedure. 
\begin{definition}[The $\tau$-censored ep-BH procedure]
\label{defi:tau-weighted-pBH}
Fix $\tau \in (0,1)$.
Let $P_1,\dots, P_K$ be p-variables for the hypotheses $H_1,\dots,H_K$ and let $E_1,\dots, E_K$ be compound e-variables. 
The $\tau$-censored ep-BH procedure rejects hypothesis $H_k$ if $P_k \leq \min\{\tau,   E_k \alpha k_{\alpha}^*/K\}$, where 
\begin{equation*} 
k_\alpha^*:=\max\left\{k\in \mathcal K\,:\,  \sum_{j \in \mathcal{K}} \id_{\{P_j \leq \min\{\tau,   E_j \alpha k/K\}\}} \geq k \right\},
\end{equation*}     
with the convention $\max(\varnothing) = 0$. 
\end{definition} 
The $\tau$-censoring construction guarantees that no p-value above $\tau$ can be rejected; see Remark~\ref{remark:tau-censoring}. The formulation above appears (without being given the name above) in~\citet{li2019multiple, ignatiadis2021covariate, zhao2024censored}.

We note that $\tau$-censoring can be implemented by directly invoking the ep-BH procedure. Specifically, suppose that for all $k \in \mathcal{K}$ the following holds:
\begin{equation}
\label{eq:tau_censoring_condition}
E_k =0\, \text{ on the event  }\, \{P_k > \tau\}.
\end{equation}
Then ep-BH is equivalent to $\tau$-censored ep-BH. Moreover, if we are given p-values $\mathbf{P}$ and compound e-values $\mathbf{E}$ that do not satisfy~\eqref{eq:tau_censoring_condition}, then we can  implement $\tau$-censored ep-BH by running ep-BH on the modified compound e-values,
$$
\tilde{E}_k := E_k \cdot  \id_{\{P_k \leq \tau\}},
$$
where here we use the convention $\infty  \times 0 = 0$.

The following theorem shows FDR control with ep-BH applied to compound e-values that satisfy~\eqref{eq:tau_censoring_condition} along with an additional leave-one-out masking condition, and its result is complementary to that of Theorem~\ref{theorem:loo_epbh}.
Below, we write $\mathbf{P} \cdot \id_{\{\mathbf{P} > \tau\}} = (P_1  \id_{\{P_1 > \tau\}},\ldots, P_K  \id_{\{P_K > \tau\}})$.

\begin{theorem}
\label{theorem:loo_epbh_censoring}
Suppose that conditional on some random variable $U$ that takes values in $\mathcal{U}$, $P_1, \ldots, P_K$ are p-values satisfying p-independence. Suppose further that
\begin{description}[style=nextline]
\item[($\ast$)]
For each $k \in \mathcal{K}$, we have $E_k = g_k(\mathbf{P} \cdot \id_{\{\mathbf{P} > \tau\}}, U)  \cdot \id_{\{P_k \leq \tau\}}$ for some  $g_k \colon [0,1]^{K} \times \mathcal{U} \to [0, \infty]$, and    moreover, $\tilde{E}_1, \ldots, \tilde{E}_K$ defined via $\tilde{E}_k = g_k(\mathbf{P}_{k \to 0} \cdot \id_{\{\mathbf{P}_{k \to 0} > \tau\}}, U)$ are compound e-values.
\end{description}
Then ep-BH applied to $\mathbf{P}$ and $\mathbf{E}$ at level $\alpha$ satisfies
$
\mathrm{FDR}_{\textnormal{ep-BH}} \leq \alpha$.
If $\tilde{E}_1, \ldots, \tilde{E}_K$
are $\varepsilon$-approximate compound e-values (instead of compound e-values), then ep-BH controls $\mathrm{FDR}$ at $\alpha(1+\varepsilon)$.
\end{theorem}
Let us parse condition ($\ast$) in more detail. 
The representation $E_k = g_k(\mathbf{P} \cdot \id_{\{\mathbf{P} > \tau\}}, U)  \cdot \id_{\{P_k \leq \tau\}}$ has two implications: first,~\eqref{eq:tau_censoring_condition} holds, and second, the compound e-values $E_k$ can only depend on $U$ and the vector of masked p-values $(\mathrm{Mask}(P_j): j \in \mathcal{K})$, where
$$
\mathrm{Mask}(P_j) = \begin{cases} P_j & \text{ if } P_j > \tau\\ 
0 & \text{ if } P_j \leq \tau.
\end{cases}
$$
In other words, under condition $(\ast)$ we are not allowed to reveal the exact magnitude of p-values smaller than $\tau$ when computing $\mathbf{E}$ from $\mathbf{P}$ (and potentially $U$).

An important upshot of Theorem~\ref{theorem:loo_epbh_censoring} over Theorem~\ref{theorem:loo_epbh} is that it does not impose a global coordinate-wise monotonicity requirement. 
We note that this theorem is similar to Lemma 1 of~\citet{zhao2024censored}, which in turn builds on earlier results of~\citet{li2019multiple,    ignatiadis2021covariate}. Also see part (ii) of Theorem 3.1 in~\citet{zijun2026minstorey} who use the terminology ``capping'' instead of ``censoring.''

\section{Omitted proofs and details for null-proportion estimators}
\label{sec:omitted_null_prop}

{
\subsection{Proof of Proposition~\ref{prop:ep-to-p}}
\label{subsec:proof_ep_to_p}

\begin{proof}
Fix $t \in (0,1)$ and $k \in \mathcal{N}$. Since $P_k$ is independent of $E_k$ and is a p-variable, we have that,
$$
\p(Q_k \leq t \mid E_k) = \p(P_k \leq t E_k \mid E_k) \leq \min\{t E_k,1\} \leq t E_k.
$$
By iterated expectation and the compound e-value property,
$$
\sum_{k \in \mathcal{N}} \p(Q_k \leq t) \leq t \sum_{k \in \mathcal{N}} \mathbb E[E_k] \leq Kt,
$$
as desired.

For a counterexample without independence, take $K=1$ under the null, with $P_1\sim\mathrm{Unif}(0,1)$ and $E_1=2\id_{\{P_1\leq1/2\}}$. Then $\mathbb E[E_1]=1$, so $E_1$ is an e-variable, but
$$
\p(Q_1\leq1/4)=\p(P_1\leq1/2)=1/2>1/4.
$$
Thus $Q_1$ is not a compound p-variable.
\end{proof}
}

\subsection{Proof of Theorem~\ref{theorem:epbh_fdr}}

\begin{proof}
The proof is similar to the proof of~\citet[Theorem 4]{ignatiadis2024evalues}. 
Since $\mathbf{E}$ is independent of $\mathbf{P}$, 
conditional on $\mathbf{E}$, the ep-BH procedure becomes a weighted p-BH procedure with weight vector 
$\mathbf{E}$ applied to the PRDS p-values $\mathbf{P}$.
Using results on the false discovery rate of the weighted p-BH procedure (e.g., \citealp[Theorem 1]{ramdas2019unified}), 
we get  
 $$
\E\left[\frac{V_{\text{ep-BH}}}{R_{\text{ep-BH}}} \;\Big |\; \mathbf{E}\right] \le \frac{1}{K} \sum_{k\in \mathcal N} E_k \alpha.
 $$ 
Hence, by iterated expectation, $\text{FDR}_{\text{ep-BH}} \leq \E[\sum_{k\in \mathcal N} E_k \alpha/K] \leq \alpha$.  
\end{proof}

\subsection{Proof of Theorem~\ref{theorem:loo_epbh}}
\begin{proof}
Write $R=R(\mathbf {P},U)$ for the number of rejections of ep-BH. Fix $k \in \mathcal{N}$. We have that:
$$
\begin{aligned}
\mathbb E\left[\frac{\id_{\{P_k \leq \alpha E_k R/K\}}}{R} \,\mid\, \mathbf{P}_{-k}, U\right]  &=  \mathbb E\left[E_k \frac{\id_{\{P_k \leq \alpha E_k R/K\}}}{E_k R} \,\mid\, \mathbf{P}_{-k}, U\right]  \\
& \leq \tilde{E}_k  \mathbb E\left[\frac{\id_{\{P_k \leq \alpha E_k R/K\}}}{E_k R} \,\mid\, \mathbf{P}_{-k}, U\right] \\ 
& \leq \tilde{E}_k \frac{\alpha}{K}.
\end{aligned}
$$
For the last inequality, we used the superuniformity lemma of~\citet[Lemma 1(b)]{ramdas2019unified} (also see Lemma 3.2 in~\citet{blanchard2008two}), noting that 
$E_k R$ is non-increasing in $P_k$ when holding $\mathbf{P}_{-k}$ and $U$ fixed. The reason is that $E_1,\ldots,E_K$ are non-increasing with $P_k$ (by assumption) and the rejections of ep-BH are non-increasing in $P_k$ and non-decreasing in $E_1,\ldots,E_K$.

Next, we can conclude using iterated expectation,
$$ \mathrm{FDR} = \sum_{k \in \mathcal{N}}\mathbb E\left[\frac{\id_{\{P_k \leq \alpha E_k R/K\}}}{R}\right] \leq \frac{\alpha}{K} \sum_{k \in \mathcal{N}}\mathbb E[\tilde{E}_k] \leq \alpha,$$
where in the last step we used the fact that $\tilde{E}_1,\ldots,\tilde{E}_K$ are compound e-values. 

By the same argument, if $\tilde{E}_1,\ldots,\tilde{E}_K$  are $\varepsilon$-approximate compound e-values, then $\mathrm{FDR}$ is controlled at $\alpha(1+\varepsilon)$.

\end{proof}

\subsection{Proof of Theorem~\ref{theorem:loo_epbh_censoring}}

\begin{proof}
As in the proof of Theorem~\ref{theorem:loo_epbh}, we write $R(\mathbf {P},U)$ for the number of rejections of ep-BH.
Fix $k \in \mathcal{N}$. 
Let us start by considering the event that $H_k$ is rejected, which is equivalent to the event that $\{ P_k \leq \alpha E_k R(\mathbf{P},U) / K\}\cap \{ P_k \leq \tau\}$. On this event, we have that $\tilde{E}_k = E_k$ and moreover that $R(\mathbf {P}_{k \to 0},U) = R(\mathbf{P},U)$.  Therefore,
$$
\begin{aligned}
\mathbb E\left[\frac{\id_{\{P_k \leq \alpha E_k R(\mathbf{P},U)/K\}\cap \{ P_k \leq \tau\}}}{R(\mathbf{P},U)} \,\mid\, \mathbf{P}_{-k}, U\right] & \leq \mathbb E\left[\frac{\id_{\{P_k \leq \alpha \tilde{E}_k R(\mathbf{P}_{k \to 0},U)/K\}}}{R(\mathbf{P}_{k \to 0},U)} \,\mid\, \mathbf{P}_{-k}, U\right] \leq \frac{\alpha}{K} \tilde{E}_k,
\end{aligned}
$$
where in the last inequality we used the fact that $P_k$ is super-uniform conditional on  $\mathbf{P}_{-k}, U$ and that $\tilde{E}_k$, $R(\mathbf{P}_{k \to 0},U)$ are only functions of  $\mathbf{P}_{-k}, U$. Then, by iterated expectation:
$$
\mathrm{FDR} \leq \frac{\alpha}{K}\sum_{k \in \mathcal{N}}\mathbb E[\tilde{E}_k],
$$
and from here we can conclude.

\end{proof}

\subsection{Proof of Theorem~\ref{proposition:generic_loo_adaptive}}

\begin{proof}
FDR control follows from Theorem~\ref{theorem:loo_epbh}.  Moreover, for all $k \in \mathcal{K}$, it holds that,
$$
\frac{1}{\hat{\pi}_0(\mathbf{P}_{k \to 0})} \geq \frac{1}{\hat{\pi}_0(\mathbf{P})},
$$
and so it follows that Algorithm~\ref{algo:epbh_null_prop} makes at least as many discoveries as Algorithm~\ref{algo:usual_null_prop}.

For strict dominance, we can apply the improvement to Storey's procedure with $\tau \in (0,1)$, also see the brief discussion after the statement of the theorem. 
\end{proof}

\subsection{Proof of Proposition~\ref{prop:new_combi}}

\begin{proof}
FDR control follows from Theorem~\ref{theorem:loo_epbh} and the fact that weighted averages of two sets of compound e-values are also compound e-values.
The compound e-values of
Algorithm~\ref{algo:epbh_combination} are at least as large as the implied compound e-values of Algorithm~\ref{algo:basic_combination}, see, e.g.,~\eqref{eq:combination-evalue-upper-bound}, and consequently, Algorithm~\ref{algo:epbh_combination} makes at least as many discoveries as Algorithm~\ref{algo:basic_combination}.  

A simple example for strict dominance takes $\lambda=1/2$ and $\hat{\pi}_0^{\mathrm{A}}(\cdot) \equiv 1$ , $\hat{\pi}_0^{\mathrm{B}}(\cdot) \equiv 2$. In this case, the two procedures reduce to p-BH at levels $2\alpha/3$ and $3\alpha/4$, respectively.
\end{proof}

\subsection{Storey procedure (Example~\ref{example:storey})}

For the Storey procedure, we still have to verify the $\varepsilon$-approximation claim when the constant $1$ is replaced by $c\in(0,1)$.  We will show the stronger claim that the Storey+ leave-one-out compound e-values with $c$ instead of $1$ are $\varepsilon$-approximate compound e-values.
Let us write
\begin{equation*}
E_k^{c} := \frac{K(1-\tau)}{c + \sum_{\ell \neq k} \id_{\{ P_{\ell} > \tau\}}}.
\end{equation*}
We also write
$\mathcal{N}_{-k} = \mathcal{N}\setminus \{k\}$ and $K_0 = \pi_0 K = |\mathcal{N}|$. Then,
$$
\begin{aligned}
E_k^{c} - E_k^1 &= \frac{K(1-\tau)(1-c)}{ \left(c+ \sum_{\ell \neq k} \id_{\{ P_{\ell} > \tau\}}\right) \left(1+ \sum_{\ell \neq k} \id_{\{ P_{\ell} > \tau\}}\right)} \\ 
& \leq \frac{2K(1-\tau)(1-c)}{c} \frac{1}{ \left( 2 + \sum_{\ell \in \mathcal{N}_{-k} } \id_{\{ P_{\ell} > \tau\}}\right) \left(1+ \sum_{\ell \in \mathcal{N}_{-k} } \id_{\{ P_{\ell} > \tau\}}\right)}
\end{aligned}
$$
Taking expectations, we find that,
$$
\mathbb E[E_k^{c} - E_k^1] \leq \frac{2K(1-\tau)(1-c)}{c K_0(K_0+1)(1-\tau)^2} = \frac{2K(1-c)}{c K_0 (K_0+1) (1-\tau)}.
$$
Above we used Lemma~\ref{lemm:second_order_binom} below (noting that by a stochastic dominance argument, for the upper bound, it suffices to assume that $P_{\ell} \sim \mathrm{Unif}(0,1)$ for all $\ell \in \mathcal{N}$).  
Hence:
$$
\begin{aligned}
\sum_{k \in \mathcal{N}} \mathbb E[E_k^c] &= \sum_{k \in \mathcal{N}}\mathbb E[E_k^1] +  \sum_{k \in \mathcal{N}}\mathbb E[E_k^{c} - E_k^1]  \leq K\left(  1 + \frac{2(1-c)}{c(K_0+1)(1-\tau)}\right).
\end{aligned}
$$

\begin{lemma}
Let $X \sim \mathrm{Binom}(n,q)$. Then:
$$
\mathbb E\left[ \frac{1}{(X+1)(X+2)}\right] \leq \frac{1}{(n+1)(n+2)q^2}.
$$
\label{lemm:second_order_binom}
\end{lemma}

\begin{proof}[Proof of Lemma~\ref{lemm:second_order_binom}]

Notice the following binomial identity:

$$
\begin{aligned}
\binom{n}{k} \frac{1}{(k+1)(k+2)} &= \frac{n!}{(n-k)! k! (k+1)(k+2)} \\ 
& = \frac{(n+2)!}{(n-k)! (k+2)!  (n+1)(n+2)} = \binom{n+2}{k+2}\frac{1}{(n+1)(n+2)}.
\end{aligned}
$$
Using this identity, we find that,
$$
\begin{aligned}
\mathbb E\left[ \frac{1}{(X+1)(X+2)}\right] &= \sum_{k=0}^n \binom{n}{k} \frac{1}{(k+1)(k+2)} q^k (1-q)^{n-k} \\
&=\frac{1}{(n+1)(n+2)q^2} \sum_{k=0}^n \binom{n+2}{k+2} q^{k+2} (1-q)^{n-k} \\ 
& \leq \frac{1}{(n+1)(n+2)q^2}.
\end{aligned}
$$
\end{proof}

\subsection{Storey+ procedure  (Proposition~\ref{proposition:Storey+})}

\begin{proof}
We omit the proof of FDR control since it's a special case of the proof of Proposition~\ref{proposition:wloostorey+} for W-LOO-Storey+.
For strict dominance, see Figure~\ref{fig:regions_alpha04_tau02} for one example. 

\end{proof}

\subsection{IBHlog procedure (Example~\ref{example:IBHlog})}

We seek to show that $E_k^{\mathrm{IBHlog}}$ defined in~\eqref{eq:Ek_IBHlog} satisfy Assumption~\ref{assumption:loo_compound} for p-independent p-values.
\citet[Appendix C]{zeisel2011fdr} prove this result under the additional assumption that null p-values are uniformly distributed, that is, when $P_k \sim \mathrm{Unif}(0,1)$ for all $k \in \mathcal{N}$. By 
a stochastic monotonicity argument, their result also implies our claim for super-uniform null p-values, i.e., when $\mathbb P[P_k \leq t] \leq t$ for all $t \in (0,1)$ and all $k \in \mathcal{N}$.

\subsection{MD+ procedure (Proposition~\ref{proposition:DM+})}

\begin{proof}
The FDR control result is a special case of the proof of Proposition~\ref{proposition:dwm+} for W-MD+. We omit the dominance argument.
\end{proof}

{
\subsection{TSTBR procedure (Example~\ref{example:TSTBR})}
\label{subsec:proof_tstbr}

Notice that
$$
R_{\mathrm{BH},\alpha'}(\mathbf{P}_{k \to 0}) \leq 1+\sum_{\ell\neq k}\id_{\{P_\ell\leq\alpha'\}},
$$
and so, using $1/(1+\alpha)=1-\alpha'$, we have that,
$$
E_k^{\mathrm{TSTBR+}}
\leq \frac{1}{1-(\alpha')^K}\cdot\frac{K(1-\alpha')}{1+\sum_{\ell\neq k}\id_{\{P_\ell>\alpha'\}}}.
$$
The right-hand side is the Storey+ compound e-value with $\tau=\alpha'$ (Proposition~\ref{proposition:Storey+}). Thus $E_k^{\mathrm{TSTBR+}}$ are also compound e-values. Since they are coordinate-wise non-increasing functions of $\mathbf{P}_{-k}$, FDR control follows from Theorem~\ref{theorem:loo_epbh}. We omit the dominance argument.
}

\subsection{TST procedure (Example~\ref{example:TST})}
\label{subsec:proof_tst}

{
We first verify the equivalence of TST and ep-BH with $E_k^{\mathrm{TST}}$, following the proof of Theorem 1 of~\citet{benjamini2006adaptive}. Write $Q_k=P_k/E_k^{\mathrm{TST}}$ and let $R_{\mathrm{TST}}$ be the number of TST discoveries.

Suppose first that $k$ is such that $H_k$ is rejected by p-BH at $\alpha'$. Since $E_k^{\mathrm{TST}} \geq 1/(1+\alpha)$, we have that,
$$
Q_k \leq (1+\alpha)P_k \leq (1+\alpha)\frac{\alpha' R_{\mathrm{BH},\alpha'}(\mathbf{P})}{K} = \frac{\alpha R_{\mathrm{BH},\alpha'}(\mathbf{P})}{K}.
$$
Thus both procedures reject all first-stage discoveries. If $R_{\mathrm{BH},\alpha'}(\mathbf{P})=K$, we can conclude.

Suppose now that $R_{\mathrm{BH},\alpha'}(\mathbf{P})<K$. Since $R_{\mathrm{BH},\alpha'}(\mathbf{P}_{k \to 1}) \leq R_{\mathrm{BH},\alpha'}(\mathbf{P})$ for all $k \in \mathcal{K}$, we have that,
$$
E_k^{\mathrm{TST}} \leq \frac{K}{(1+\alpha)(K-R_{\mathrm{BH},\alpha'}(\mathbf{P}))}.
$$
Hence ep-BH cannot make any discovery not made by p-BH at level $\alpha' K/(K-R_{\mathrm{BH},\alpha'}(\mathbf{P}))$, the second stage of TST. Conversely, for a hypothesis not rejected at the first stage, replacing $P_k$ by $1$ leaves $R_{\mathrm{BH},\alpha'}(\mathbf{P})$ unchanged, so the above bound on $E_k^{\mathrm{TST}}$ is an equality. For such a TST discovery, the second-stage BH threshold gives
$$
Q_k = \frac{(1+\alpha)(K-R_{\mathrm{BH},\alpha'}(\mathbf{P}))}{K}P_k
\leq \frac{(1+\alpha)\alpha' R_{\mathrm{TST}}}{K}
= \frac{\alpha R_{\mathrm{TST}}}{K}.
$$
Together with the earlier bound for first-stage discoveries, this shows that all $R_{\mathrm{TST}}$ TST discoveries satisfy $Q_k\leq\alpha R_{\mathrm{TST}}/K$. Thus ep-BH rejects all TST discoveries.
}

The proof that~\eqref{eq:E_k_tst} are compound e-values is as follows. Let $\alpha_K = \alpha'(K-1)/K$ and notice that

$$R_{\mathrm{BH},\alpha'}(\mathbf{P}_{{k} \to 1}) \leq \sum_{\ell \neq k } \id_{\{P_{\ell} \leq \alpha_K \}},$$
and so,
$$
K-R_{\mathrm{BH},\alpha'}(\mathbf{P}_{{k} \to 1})  \geq 1 + \sum_{\ell \neq k } \id_{\{P_{\ell} > \alpha_K \}}.
$$
Then:
$$
\begin{aligned}
\sum_{k \in \mathcal{N}} \mathbb E[E_k^{\mathrm{TST}}] &\leq \frac{1}{1+\alpha}\sum_{k \in \mathcal{N}}\mathbb E\left[\frac{K}{ 1+ \sum_{\ell \neq k } \id_{\{P_{\ell} > \alpha_K  \}} }\right]\\ 
& = \frac{1}{(1+\alpha)(1-\alpha_K)}\sum_{k \in \mathcal{N}}\mathbb E\left[\frac{K(1-\alpha_K)}{ 1+ \sum_{\ell \neq k } \id_{\{P_{\ell} > \alpha_K  \}} }\right] \\ 
& \leq \frac{K}{(1+\alpha)(1-\alpha_K)} \\ 
& = K\frac{K}{K+\alpha}.
\end{aligned}
$$
In the last step we used the fact that
$$
1-\alpha_K = 1-\frac{\alpha (K-1)}{(1+\alpha)K} = \frac{K + \alpha}{(1+\alpha)K}.
$$
It follows that $E_k^{\mathrm{TST}}$ are compound e-values, and moreover they remain compound e-values if multiplied by $(K+\alpha)/K$. Thus $E_k^{\mathrm{TST+}}$ are also compound e-values and the resulting ep-BH procedure dominates TST.

{
To prove FDR control for TST and TST+, note that both choices of e-values are coordinate-wise non-increasing functions of $\mathbf{P}_{-k}$. Hence condition ($\ast$) of Theorem~\ref{theorem:loo_epbh} holds with $\tilde{E}_k=E_k$, and $\mathrm{FDR}\leq\alpha$ follows.
}

{
\subsection{Comparison of TSTBR+ and TST+}
\label{subsec:comparison_tst}

The dominance of TST+ over TST was shown in Appendix~\ref{subsec:proof_tst}. 
TSTBR+ dominates TST since $R_{\mathrm{BH},\alpha'}(\mathbf{P}_{k\to0})\geq 1+R_{\mathrm{BH},\alpha'}(\mathbf{P}_{k\to1})$ implies $E_k^{\mathrm{TSTBR+}}\geq E_k^{\mathrm{TST}}$ for every $k$, with strict dominance in the first example below.

To see that neither improvement dominates the other, take $K=2$ and $\alpha=0.05$, so that $\alpha'=1/21$. If $\mathbf{P}=(0.04,0.06)$, then,
$$
\bigl(E_1^{\mathrm{TSTBR+}},E_2^{\mathrm{TSTBR+}}\bigr)=\left(\frac{21}{22},\frac{21}{11}\right),
\;\;
E_1^{\mathrm{TST+}}=E_2^{\mathrm{TST+}}=\frac{41}{42}.
$$
Both compound-e-weighted p-values for TSTBR+ are below $\alpha$, so it rejects both hypotheses. For TST+, the smaller compound-e-weighted p-value exceeds $\alpha/2$ and the larger exceeds $\alpha$, so it makes no discoveries.

Conversely, if $\mathbf{P}=(0.048,0.048)$, then $E_1^{\mathrm{TSTBR+}}=E_2^{\mathrm{TSTBR+}}=21/22$ and $E_1^{\mathrm{TST+}}=E_2^{\mathrm{TST+}}=41/42$. The compound-e-weighted p-values exceed $\alpha$ for TSTBR+ and are below $\alpha$ for TST+. Thus TSTBR+ makes no discoveries, while TST+ rejects both hypotheses.
}

\subsection{MABH procedure (Example~\ref{example:mabh})}
\label{subsec:proof_mabh}

Here we verify that the \smash{$E_k^{\mathrm{MABH}}$} in~\eqref{eq:mabh} are indeed compound e-values and that ep-BH with these compound e-values is identical to MABH. We then recall the argument of~\citet{solari2017minimally} for FDR control under PRDS.

We split cases according to the value of $\pi_0$. If $\pi_0=1$ (that is, we are under the global null), then $\mathbb P(R_{\mathrm{BH},\alpha}(\mathbf{P})>0) \leq \alpha$ and so,
$$
\sum_{k \in \mathcal{K}} \mathbb E[E_k^{\mathrm{MABH}}] \leq K   \frac{K}{K-1}\alpha \leq K,
$$
where we use the fact that $\alpha \leq (K-1)/K$.
Meanwhile, if $\pi_0 < 1$, then it must be the case that $|\mathcal{N}| \leq K-1$. Then:
$$
\sum_{k \in \mathcal{N}} \mathbb E[E_k^{\mathrm{MABH}}]  \leq (K-1) \frac{K}{K-1} = K. 
$$
Thus \smash{$E_k^{\mathrm{MABH}}$} are indeed compound e-values.

Next we show the equivalence to MABH. If $R_{\mathrm{BH},\alpha}(\mathbf{P})=0$, then both p-BH and MABH make no discoveries. Similarly, in this case, \smash{$E_k^{\mathrm{MABH}}=0$} and so ep-BH also makes no discoveries.
If $R_{\mathrm{BH},\alpha}(\mathbf{P})>0$, then MABH is equivalent to p-BH applied at level $\alpha K/(K-1)$, which is equivalent to ep-BH with e-values equal to $K/(K-1)$.

{
To explain
the FDR control proof of~\citet{solari2017minimally}, we again split cases according to $\pi_0$. Under the global null, the FDR of MABH equals $\mathbb P(R_{\mathrm{BH},\alpha}(\mathbf{P})>0) \leq \alpha$. If $\pi_0<1$, MABH either makes no discoveries or agrees with p-BH at level $\alpha K/(K-1)$. Hence its false discovery proportion is bounded by that of this p-BH procedure, whose FDR under PRDS is at most $\alpha|\mathcal{N}|/(K-1) \leq \alpha$.
}

\subsection{W-LOO-Storey+ (Proposition~\ref{proposition:wloostorey+})}

\begin{proof}[Proof of Proposition~\ref{proposition:wloostorey+}]

Since each $E_k$ is a coordinate-wise non-increasing function of
$\mathbf{P}_{-k}$, condition ($\ast$) of
Theorem~\ref{theorem:loo_epbh} holds with $\tilde{E}_k=E_k$ once the
compound e-value property is verified. It therefore suffices to check
that the $E_k$ are compound e-values.

Let
$$
\mathcal{N}_+ := \{k\in\mathcal{N}:w_k>0\}.
$$
Under the convention $0/0=0$, we have $E_k=0$ for
$k\in\mathcal{N}\setminus\mathcal{N}_+$. If
$\mathcal{N}_+=\varnothing$, the claim is immediate. Otherwise, let
$$
(B_k)_{k\in\mathcal{N}_+}
\simiid \mathrm{Bernoulli}(1-\tau).
$$
By super-uniformity and p-independence,
$(\id_{\{P_\ell>\tau\}})_{\ell\in\mathcal{N}_+}$ jointly
stochastically dominates $(B_\ell)_{\ell\in\mathcal{N}_+}$. We obtain
$$
\begin{aligned}
\sum_{k\in\mathcal{N}}
\mathbb{E}\left[E_k\right]
&=
\frac{K}{1-\tau^K}
\sum_{k\in\mathcal{N}_+}
\mathbb{E}\left[
\frac{w_k(1-\tau)}
{w_k+\sum_{\ell\neq k}
w_\ell\id_{\{P_\ell>\tau\}}}
\right]
\\
&\leq
\frac{K}{1-\tau^K}
\sum_{k\in\mathcal{N}_+}
\mathbb{E}\left[
\frac{w_k(1-\tau)}
{w_k+\sum_{\ell\in\mathcal{N}_+\setminus\{k\}}
w_\ell\id_{\{P_\ell>\tau\}}}
\right]
\\
&\leq
\frac{K}{1-\tau^K}
\sum_{k\in\mathcal{N}_+}
\mathbb{E}\left[
\frac{w_k(1-\tau)}
{w_k+\sum_{\ell\in\mathcal{N}_+\setminus\{k\}}
w_\ell B_\ell}
\right]
\\
&=
\frac{K}{1-\tau^K}
\sum_{k\in\mathcal{N}_+}
\mathbb{E}\left[
\frac{w_kB_k}
{\sum_{\ell\in\mathcal{N}_+}w_\ell B_\ell}
\right]
\\
&=
\frac{K}{1-\tau^K}
\mathbb{P}\left(
\sum_{\ell\in\mathcal{N}_+}w_\ell B_\ell>0
\right)
\\
&=
K\frac{1-\tau^{|\mathcal{N}_+|}}{1-\tau^K}
\leq K.
\end{aligned}
$$
The first inequality follows by dropping the nonnegative non-null
terms from the denominators. The second inequality follows from the
joint stochastic dominance above, since each summand is
coordinate-wise non-increasing in the null indicators. The subsequent
equality uses $\mathbb{E}[B_k]=1-\tau$, the independence of $B_k$ from
$(B_\ell)_{\ell\neq k}$, the fact that $B_k\in\{0,1\}$ and the convention $0/0=0$. The final
inequality follows from $|\mathcal{N}_+|\leq K$ and $\tau\in(0,1)$.
Thus, the $E_k$ are compound e-values, and the FDR guarantee follows
from Theorem~\ref{theorem:loo_epbh}.
The dominance argument is omitted. 

\end{proof}

\section{Simultaneous t-tests}
\label{sec:simultaneous_ttests}

\subsection{Statistical setting and proposed methods}
\label{subsec:statistical_setting_ttest}
In this section, we apply the leave-one-out ep-BH framework to the problem of simultaneous t-tests.
For the $k$-th hypothesis we observe 
\begin{equation}
Y_{kj} \sim \mathrm{N}(\mu_k,\, \sigma_k^2), \text{ for } j=1,\dotsc,n,\;\; \mu_k \in \mathbb R,\;\; \sigma_k > 0,
\label{eq:gaussian_replicates}
\end{equation}
and we assume that all $Y_{kj},\, 1 \leq k \leq K,\, 1 \leq j \leq n$ are mutually independent with $n \geq 2$. We seek to test $H_k: \mu_k=0$ and to do so, we compute p-values based on the standard t-test:
\begin{equation}
\label{eq:one_sample_ttest_pvalue}
\hat{\mu}_k := \frac{1}{n}\sum_{j=1}^n Y_{kj},\quad \hat{\sigma}_k^2 := \frac{1}{n-1} \sum_{j=1}^n (Y_{kj} - \hat{\mu}_k)^2,\quad T_k := \frac{\sqrt{n}\hat{\mu}_k}{\hat{\sigma}_k},
\end{equation}
and $P_k := 2\{1-F_{t, n-1}(|T_k|)\}$, where $F_{t, n-1}$ is the CDF of the t-distribution with $n-1$ degrees of freedom. Note that the p-values $P_1,\ldots,P_K$ satisfy p-independence.

An important idea that goes back to~\citet{westfall2004weighted} is that the above p-values may leave out important information. Consider the mean of squares for the $k$-th hypothesis,
$$S_k^2 := \frac{1}{n}\sum_{j=1}^{n} Y_{kj}^2,\;\;\; S_k^2 \sim \frac{\sigma_k^2}{n} \chi^2_n\left( \frac{n \mu_k^2}{\sigma_k^2}\right),$$
where $\chi^2_n(u)$ is the chi-square distribution with $n$ degrees of freedom and noncentrality parameter $u$. From this expression we see that $S_k^2$ contains information about $\mu_k$. On the other hand, by Basu's theorem, it holds that
\begin{equation}
\label{eq:s_k_p_k_independence}
S_k^2\, \text{ is independent of}\; P_k \;\text{ for } k \in \mathcal{N}.
\end{equation}
{
Together with independence across hypotheses, \eqref{eq:s_k_p_k_independence} implies that conditionally on $(S_1^2,\ldots,S_K^2)$, $P_1,\ldots,P_K$ are p-values that satisfy p-independence.
}
The above suggests the following procedure. Let $\psi: [0,\infty) \to [0, \infty)$ be any non-decreasing function and define the normalized weights,
\begin{equation}
w_k := \frac{K\psi(S_k^2)}{\sum_{\ell \in \mathcal{K}}  \psi(S_{\ell}^2)},
\label{eq:normalized_weights}
\end{equation}
where we would set $w_k=1$ if the denominator is zero.
{
Then both weighted p-BH and W-Max-Storey (Example~\ref{example:max_storey}) with p-values $P_k$ and weights $w_k$ control the FDR. For either procedure, let $V$ and $R$ denote the numbers of false discoveries and total discoveries, respectively. Conditional on $(S_1^2,\ldots,S_K^2)$, the weights are deterministic and $P_1,\ldots,P_K$ are p-values that satisfy p-independence, so the FDR bounds for deterministic weights give $\mathbb E[V/(R\lor 1)\mid S_1^2,\ldots,S_K^2]\leq\alpha$. By iterated expectation, $\mathrm{FDR}=\mathbb E[V/(R\lor 1)]\leq\alpha$.
}

Various weighting functions have been proposed: $\psi(s) = s^{\nu}$ for fixed $\nu > 0$ \citep{westfall2004weighted}, $\psi(s) = \id(s > c)$ for $c > 0$ \citep{finos2007fdr, bourgon2010independent, guo2017analysis}, and $\psi(s) = s \id(s>c)$ for $c > 0$ \citep{finos2007fdr}.

\citet{ignatiadis2024evalues} observe that one could in principle develop more powerful procedures than the above by using compound e-values. As an example, consider the case with $\psi(u)=u$. Then, instead of the weights in~\eqref{eq:normalized_weights}, an oracle could use the (oracle) compound e-values,
$$
E_k^{\mathrm{or}} := \frac{K S_k^2}{ \sum_{\ell \in \mathcal{K}} \sigma_{\ell}^2} \quad\quad \text{ as compared to }\quad\quad w_k := \frac{K S_k^2}{\sum_{\ell \in \mathcal{K}} S_{\ell}^2},
$$
and the former are preferable since
$$\sum_{\ell \in \mathcal{K}} S_{\ell}^2 \approx \mathbb E\left[\sum_{\ell \in \mathcal{K}} S_{\ell}^2\right] = \sum_{\ell \in \mathcal{K}}( \sigma_{\ell}^2 + \mu_{\ell}^2) \geq \sum_{\ell \in \mathcal{K}} \sigma_{\ell}^2.
$$
\citet{ignatiadis2024evalues} propose to approximate the oracle $E_k^{\mathrm{or}}$ by $E_k^{\mathrm{approx}} := K S_k^2 / \sum_{\ell=1}^K \hat{\sigma}_{\ell}^2.$ Using concentration arguments that \smash{$\sum_{\ell=1}^K \hat{\sigma}_{\ell}^2 \approx \sum_{\ell \in \mathcal{K}} \sigma_{\ell}^2$} with high probability, they then establish that ep-BH with the t-test p-values and the (asymptotic) compound e-values \smash{$E_k^{\mathrm{approx}}$} asymptotically controls the FDR under certain conditions. 

Their procedure thus provides a strong proof-of-principle that one can go beyond the normalization in~\eqref{eq:normalized_weights}. However, it leaves open the following question. How to develop a procedure that improves upon weighted p-BH with the weights in~\eqref{eq:normalized_weights},
has finite-sample FDR control (without requiring any further assumption) and that can work for any choice of the function $\psi(\cdot)$?

To address the above question, we propose the LOO-Var+ method. Let 
$$B_k \simiid \frac{n}{n-1}\mathrm{Beta}( (n-1)/2,\, 1/2)$$ for all $k \in \mathcal{K}$ and independent of everything else. For $k\in\mathcal N$, $B_kS_k^2$ has the same distribution as $\hat\sigma_k^2$ but is independent of $P_k$. Then define,
\begin{equation}
E_k^{\mathrm{LOOVar+}} := \frac{ K\psi( B_k S_k^2)}{ \psi( B_k S_k^2) +  \sum_{\ell \neq k} \psi( \hat{\sigma}_{\ell}^2)}, \,\, \text{ or }\,\, E_k^{\mathrm{LOOVar+}} := \int \frac{ K\psi( b S_k^2)}{ \psi( b S_k^2) +  \sum_{\ell \neq k} \psi( \hat{\sigma}_{\ell}^2)} p(b)\dd b,
\label{eq:loovar_compound_evals}
\end{equation}
with the convention $0/0=0$ and where $p(\cdot)$ is the density of $B_k$.\footnote{A related construction using Beta random variables appears in~\citet[Appendix D.3]{barber2025false} to construct approximate compound p-values in our setting.} (The right-hand option is the derandomized version of the left-hand side.)

We have the following result.
\begin{theorem}
Suppose that $\psi(\cdot): [0,\infty) \to [0,\infty)$ is any fixed non-decreasing function. Then,
$E_1,\ldots,E_K$ defined in~\eqref{eq:loovar_compound_evals} are compound e-values. The ep-BH procedure with the t-test p-values $P_1,\ldots,P_K$ and these compound e-values controls the $\mathrm{FDR}$ at level $\alpha$. 
\label{theorem:loo_var}
\end{theorem}

In this setting, we could also have treated the normalized weights in~\eqref{eq:normalized_weights} as deterministic weights and then combined them with a weighted adaptive procedure such as those in Section~\ref{subsec:wtd_adaptivity}. Below, we briefly compare W-LOO-Storey+ vs.\ LOO-Var+ vs.\ normalized weights (without null-proportion adaptivity) in the case of $\psi(s) =s$:
$$
E_k^{\mathrm{LOOVar+}} = \int \frac{K b S_k^2}{bS_k^2 + \sum_{\ell \neq k} \hat{\sigma}_{\ell}^2} p(b) \dd b,\;\;\,E_k^{\mathrm{WStorey+}} = \frac{K S_k^2 (1-\tau)/(1-\tau^K)}{S_k^2 + \sum_{\ell \neq k} S_{\ell}^2\id_{\{ P_{\ell} > \tau\}}},\;\;\,w_k = \frac{K S_k^2}{\sum_{\ell \in \mathcal{K}} S_{\ell}^2} 
$$
We first consider a sparse regime with strong signals. Specifically, we fix the number of alternatives $K_1$ and the number of hypotheses $K$, and suppose that $\mu_k \equiv \mu$ for all $k \notin \mathcal{N}$. We then take $\mu \to \infty$. We find the following limiting expressions for $k \notin \mathcal{N}$:
$$
E_k^{\mathrm{LOOVar+}} \to K,\;\;\; E_k^{\mathrm{WStorey+}} \to K\frac{1-\tau}{1-\tau^K},\;\;\; w_k \to \frac{K}{K_1}.
$$
Thus, in this sparse regime with strong signals, $E_k^{\mathrm{LOOVar+}}$ outperforms both $E_k^{\mathrm{WStorey+}}$ and normalized weights. Moreover, if there is only one signal $(K_1=1)$, then the normalized weights $w_k$ outperform W-LOO-Storey+.

We next consider regimes for which $K$ is large, the denominators of the compound e-values/weights concentrate around their expectation and the leave-one-out contribution is negligible:
$$
\begin{aligned}
&E_k^{\mathrm{LOOVar+}} \approx \frac{K S_k^2}{ \sum_{\ell \in \mathcal{N}}\sigma_{\ell}^2 \,+\, \sum_{\ell \notin \mathcal{N}}\sigma_{\ell}^2},\;\;\; E_k^{\mathrm{WStorey+}} \approx \frac{K S_k^2}{ \sum_{\ell \in \mathcal{N}} \sigma_{\ell}^2 + \sum_{\ell \notin \mathcal{N}} \mathbb E\left[ S_{\ell}^2\frac{\id_{\{ P_{\ell} > \tau\}}}{1-\tau}\right]},\\ 
&w_k \approx \frac{KS_k^2}{ \sum_{\ell \in \mathcal{N}}\sigma_{\ell}^2 \,+\, \sum_{\ell \notin \mathcal{N}}(\sigma_{\ell}^2 + \mu_{\ell}^2)}.
\end{aligned}
$$
Therefore the comparison in this regime boils down to comparing (the smaller, the better) 
$$
\sum_{\ell \notin \mathcal{N}}\sigma_{\ell}^2 \;\;\text{ vs. }\;\; \sum_{\ell \notin \mathcal{N}} \mathbb E\left[S_{\ell}^2\frac{\id_{\{ P_{\ell} > \tau\}}}{1-\tau} \right] \;\;\text{ vs. } \sum_{\ell \notin \mathcal{N}}(\sigma_{\ell}^2 + \mu_{\ell}^2).
$$
We see that LOO-Var+ eliminates $\mu_{\ell}$ for $\ell \notin \mathcal{N}$. By contrast, W-LOO-Storey+ may still be impacted by $\mu_{\ell}$ since the distribution of $S_{\ell}$ for $\ell \notin \mathcal{N}$ depends on $\mu_{\ell}$. However, LOO-Var+ always pays the price $\sigma_{\ell}^2$ for all $\ell \notin \mathcal{N}$, while W-LOO-Storey+ could potentially pay a price $\approx 0$ if $|\mu_{\ell}|/\sigma_{\ell}$ is sufficiently large so that $P_{\ell} > \tau$ with very low probability. In this sense, W-LOO-Storey+ can outperform LOO-Var+ when signals are strong and $\pi_0 \ll 1$.

Overall, we may expect both LOO-Var+ and W-LOO-Storey+ to outperform each other in different regimes; see also the simulation study below.

\subsection{Proof of Theorem~\ref{theorem:loo_var}}

\begin{proof}
We first prove that the randomized $E_k$ are compound e-values; the statement for the derandomized $E_k$ then follows immediately.
The argument hinges on the following observation. For $k \in \mathcal{N}$,
$$
(n-1)\hat{\sigma}_k^2 \sim \sigma_k^2 \chi^2_{n-1},\;\; n S_k^2 \sim \sigma_k^2 \chi^2_n.   
$$
Since $\hat{\sigma}_k^2 = (\hat{\sigma}_k^2/S_k^2) \cdot S_k^2$ with  $\hat{\sigma}_k^2/S_k^2$ independent of $S_k^2$, this means that $\hat{\sigma}_k^2 \stackrel{\mathcal{D}}{=}B_k S_k^2$ and so:
$$
\mathbb E[E_k] =K \mathbb E\left[ \frac{ \psi( B_k S_k^2)}{ \psi( B_k S_k^2) +  \sum_{\ell \neq k} \psi( \hat{\sigma}_{\ell}^2)}\right] = K\mathbb E\left[ \frac{ \psi( \hat{\sigma}_k^2)}{  \sum_{\ell \in \mathcal{K}} \psi( \hat{\sigma}_{\ell}^2)}\right]. 
$$
Summing over all $k \in \mathcal{N}$, it follows that $E_1,\ldots,E_K$ are compound e-values.

To prove FDR control, we will apply  Theorem~\ref{theorem:loo_epbh} with condition ($\ast$) and the choice $U := (S_1^2, \ldots, S_K^2, B_1,\ldots,B_K)$ for the randomized procedure and $U:= (S_1^2, \ldots, S_K^2)$ for the derandomized procedure. In particular, we must argue about the monotonicity of the $E_k$ when treated as functions of $P_1,\ldots,P_K$ and $U$. To this end, it will be helpful to note the following standard decomposition: 
$$\sum_{j=1}^n Y_{kj}^2 = \sum_{j=1}^n (Y_{kj} - \hat{\mu}_k)^2 + n \hat{\mu}_k^2\;\; \Longrightarrow\;\; n S_k^2 = (n-1)\hat{\sigma}_k^2 + n \hat{\mu}_k^2.$$
Dividing by $\hat{\sigma}_k^2$ yields that $n S_k^2 / \hat{\sigma}_k^2 = (n-1) + T_k^2$, i.e., 
$$\hat{\sigma}_k^2 = \frac{n S_k^2}{(n-1) + T_k^2}.$$
Thus $\hat{\sigma}_k^2$ is decreasing in 
$T_k^2$, and so it is increasing in $P_k$. Meanwhile, $E_k$ is non-increasing in all $\hat{\sigma}_{\ell}^2$ and so it is also non-increasing in all $P_1,\ldots,P_K$.
\end{proof}

\subsection{Simulation study}
\label{sec:simulation}

We consider the t-test setting with $K=200$ and $n \in \{2,5,20\}$. We fix the number of alternative hypotheses $K_1 \in \{2, 5, 100\}$. We set $\sigma_k^2=1$ for all $k$ and for the alternatives, $\mu_k = \xi \sqrt{n}$ with effect size $\xi \in [2,8]$. 

We apply the following procedures to control the FDR at $\alpha=0.1$. For weighted methods, we consider two choices of the function $\psi(\cdot)$, namely $\psi(u) = u^4$ and $\psi(u) = u^2 \id_{\{u \geq 1 \}}$.

\begin{enumerate}[leftmargin=*]
\item p-BH (BH) based on the t-test p-values;
\item weighted BH (W-BH) with weights normalized to sum to $K$;
\item weighted Storey as in Example~\ref{example:max_storey} (W-Max-Storey) with $\tau=0.5$;
\item W-LOO-Storey as in Example~\ref{example:loo_weighted_storey} with $\tau=0.5$ and $\tau$-censoring;
\item W-LOO-Storey+ as described in Proposition~\ref{proposition:wloostorey+};
\item the leave-one-out sample variance method from Appendix~\ref{sec:simultaneous_ttests} (LOO-Var+);
\item a method that averages the implicit compound e-values of LOO-Var+ and W-LOO-Storey+ following Algorithm~\ref{algo:epbh_combination} (LOO-Var/Storey+).
\end{enumerate}
In each case, we report the power (defined as the expected fraction of alternatives discovered) and FDR, averaged over 10,000 Monte Carlo replications.\footnote{Figure-generation code was developed with assistance from GPT 5.3 and reviewed by the authors. We provide code to reproduce this simulation study under the following GitHub repository:\\ \url{https://github.com/nignatiadis/pi0-and-compound-evalues-paper}}  

We show results in Figure~\ref{fig:psi1} for $\psi(u)=u^4$. For readability, we omit W-LOO-Storey~\citep{zhao2024censored} from the figure and discuss its comparison with the uniformly improved W-LOO-Storey+ separately below.
We note that all methods control FDR at the nominal level.
In most settings, BH and W-BH have the least power, although W-BH outperforms W-Max-Storey for $K_1=n=2$. LOO-Var+ and W-LOO-Storey+ are in general the most powerful procedures, and are also tied in most settings. The largest differences between them occur for $K_1=n=2$, in which case LOO-Var+ is more powerful, and for $K_1=100$ (i.e., $\pi_0=0.5$) and $n\in \{5,20\}$, in which case W-LOO-Storey+ is slightly more powerful. These observations are in line with our heuristic power comparison between LOO-Var+ and W-LOO-Storey+ in Section~\ref{subsec:statistical_setting_ttest}. The combination method (LOO-Var/Storey+) usually has performance that is sandwiched between the performance of W-LOO-Storey+ and LOO-Var+, and so provides a promising method when it is a priori unclear if 
LOO-Var+ or W-LOO-Storey+ will be more powerful. Figure~\ref{fig:psi2} reports results for $\psi(u)=u^2 \id_{\{u \geq 1 \}}$; the overall pattern is broadly similar.

Finally, we compare W-LOO-Storey+ with W-LOO-Storey. By construction, W-LOO-Storey+ always makes at least as many discoveries as W-LOO-Storey. In most simulation scenarios, the differences are negligible, which is unsurprising given that we expect the improvement to be small. However, for $K_1=100$ and $n\in\{2,5\}$, the power gain is noticeable; we display it in Figure~\ref{fig:W-LOO-Storey}. In these settings, some hypotheses with $P_i>\tau$ can still be rejected by W-LOO-Storey+ because their corresponding $E_i$ are sufficiently large, whereas such hypotheses are excluded by the $\tau$-censoring of W-LOO-Storey.

\begin{figure}
\centering
\includegraphics[width=0.95\textwidth]{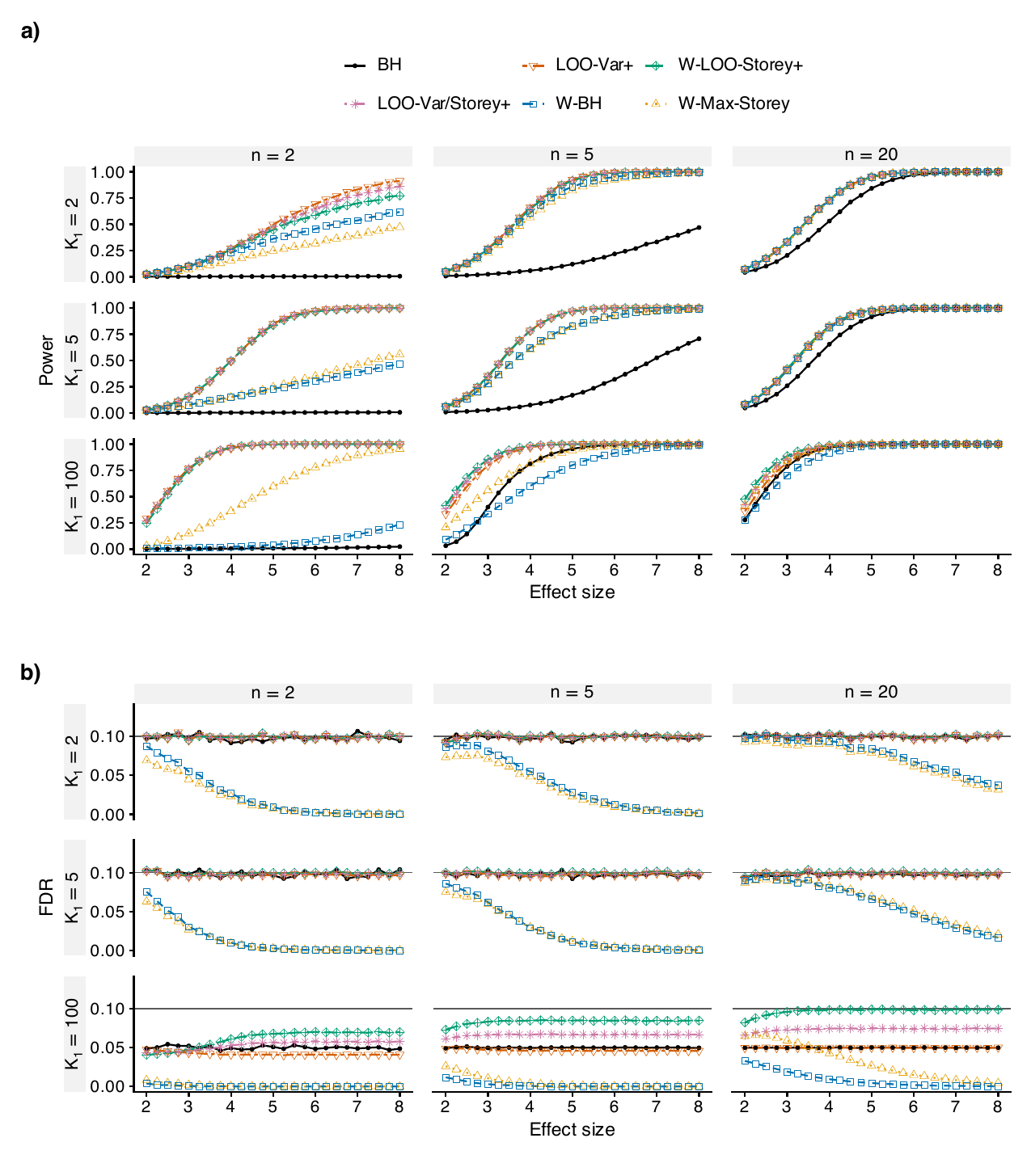}
\caption{ Power (\textbf{a}) and FDR (\textbf{b}) versus effect size $\xi$ for 6 different methods for multiple testing with simultaneous t-tests. We use the weighting function $\psi(u)=u^4$. The different facets correspond to different sample sizes ($n \in \{2,5,20\}$) and numbers of alternative hypotheses $K_1 \in \{2,5,100\}$, where the total number of hypotheses is $K=200$. All methods control FDR at the nominal $\alpha=0.1$. LOO-Var+ (ours) and W-LOO-Storey+ (ours) have the most power across facets.}
\label{fig:psi1}
\end{figure}

\begin{figure}
\centering
\includegraphics[width=0.95\textwidth]{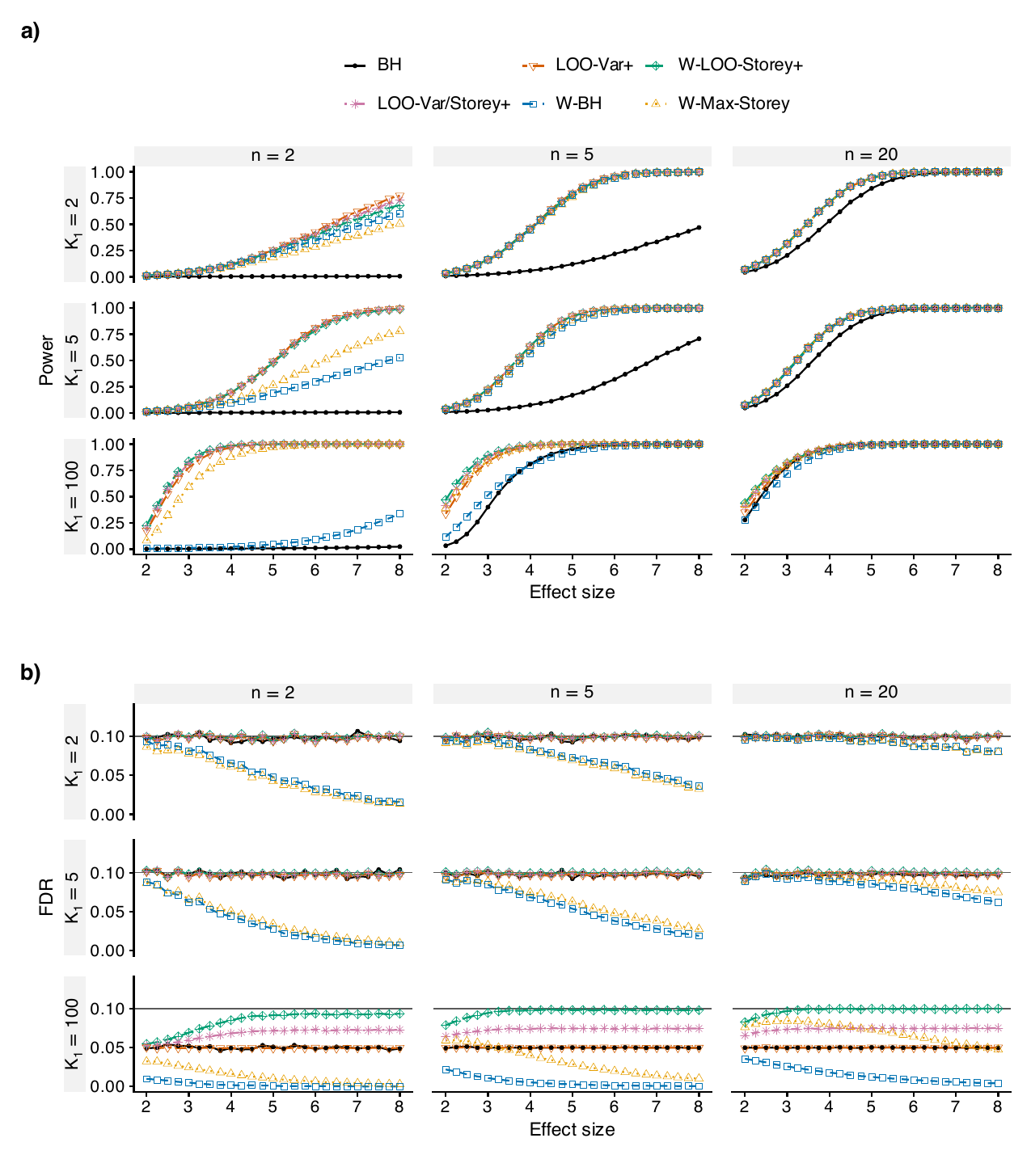}
\caption{Analogous to Figure~\ref{fig:psi1} with a different choice of weighting function, namely  $\psi(u)=u^2 \id_{\{u \geq 1 \}}$.}
\label{fig:psi2}
\end{figure}

\begin{figure}
    \centering
    \begin{tabular}{ll}
    \textbf{a)} $\psi(u) = u^4$ & \textbf{b)} $\psi(u) =u^2 \id_{\{u \geq 1 \}}$ \\
        \includegraphics[width=0.49\textwidth]{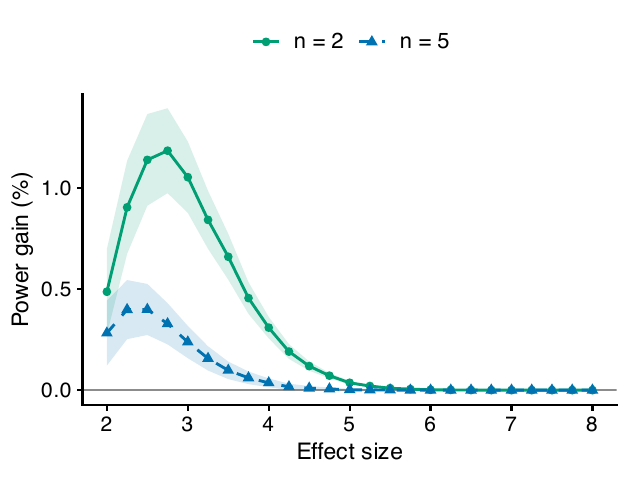} &
        \includegraphics[width=0.49\textwidth]{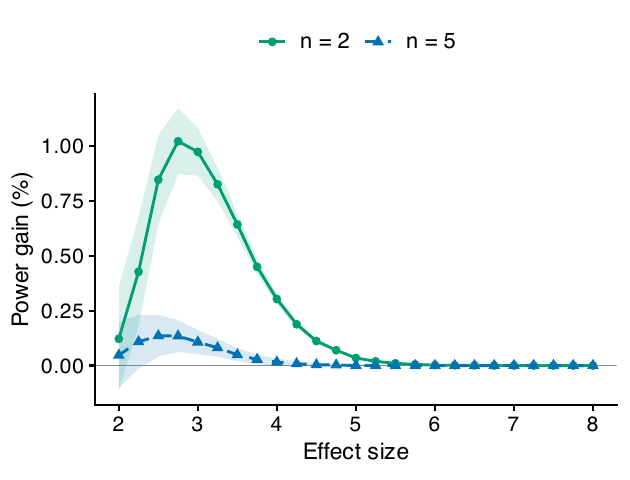}
    \end{tabular}
\caption{Power gain (\%) of W-LOO-Storey+ over W-LOO-Storey ($\tau$-censored) for $K_1=100$, as a function of effect size $\xi$. Bands denote $\pm 1$ standard error of the estimated power gain. Panels correspond to (\textbf{a}) $\psi(u)=u^4$ and (\textbf{b}) $\psi(u)=u^2\id_{{u\geq1}}$; curves correspond to $n=2$ and $n=5$.}
    \label{fig:W-LOO-Storey}
\end{figure}

\end{document}